\documentclass[journal]{IEEEtran}
\ifCLASSINFOpdf
\else
\fi
\usepackage{graphicx}
\usepackage{setspace}
\usepackage{subfigure}
\usepackage{subcaption}
\usepackage{amsmath}
\usepackage{amssymb} 
\usepackage{mathtools} 
\usepackage{comment}
\usepackage{bm} 
\usepackage{amsthm}
\usepackage{cite}
\usepackage{float}
\usepackage[usenames,dvipsnames]{pstricks}
\usepackage{epsfig}
\usepackage{pst-grad} 
\usepackage{pst-plot} 
\usepackage{color}
\usepackage[usenames,dvipsnames]{xcolor}
\newtheorem{remark}{Remark}

\newtheorem{corollary}{Corollary}
\newtheorem{lemma}{Lemma}

\usepackage{tabularx}
\usepackage[T1]{fontenc}
\newcommand{\lb} {\left}
\newcommand{\rb} {\right}
\newcommand{\nn} {\nonumber}

\newcommand{\norm}[1]{\left\lVert#1\right\rVert}
\newcommand{\abs}[1]{\left\lvert#1\right\rvert}

\allowdisplaybreaks
\newcommand{\diag}{\operatorname{diag}}

\newcommand{\cov}{\operatorname{cov}}

\begin{document}
\title{Jointly Optimal RIS Placement and Power Allocation for Underlay D2D Communications:\\ An Outage Probability Minimization Approach}
\author
{Sarbani Ghose,~\IEEEmembership{Member,~IEEE},
Deepak Mishra,~\IEEEmembership{Senior Member,~IEEE},\\
Santi P. Maity,~\IEEEmembership{Member, IEEE}, and George C. Alexandropoulos,~\IEEEmembership{Senior Member, IEEE}
\thanks{S. Ghose is with DSZ Innovation Laboratories Private Limited, Kolkata, India, (email:sarbani.ghose@dszlabs.com).}
\thanks{D. Mishra is with the School of Electrical Engineering and Telecommunications at the University of New 
South Wales (UNSW) Sydney, NSW 2052, Australia, (email: d.mishra@unsw.edu.au).}
\thanks{S. P. Maity is with the Department of Information Technology, Indian Institute of Engineering Science and Technology, 
Shibpur, Howrah, West Bengal, 711103, India, (email: santipmaity@it.iiests.ac.in).}
\thanks{G. C. Alexandropoulos is with the Department of Informatics and Telecommunications, 
National and Kapodistrian University of Athens, Greece, (email: alexandg@di.uoa.gr).}
\thanks{The work of Prof. Alexandropoulos has been supported by the EU H2020 RISE-6G project under grant number 101017011.}}
\maketitle
\begin{abstract}
In this paper, we study underlay device-to-device (D2D) communication systems empowered by a reconfigurable 
intelligent surface (RIS) for cognitive cellular networks. Considering Rayleigh fading channels and the general 
case where there exist both the direct and RIS-enabled D2D channels, the outage probability (OP) of 
the D2D communication link is presented in closed-form. 
Next, for the considered RIS-empowered underlaid D2D system, we frame an OP minimization problem. 
We target the joint optimization of the transmit power at the D2D 
source and the RIS placement, under constraints on the transmit power at the D2D source and 
on the limited interference imposed on the cellular user for two RIS deployment topologies. 
Due to the coupled optimization variables, the formulated optimization problem is extremely intractable. 
We propose an equivalent transformation which we are able to solve analytically. In the transformed problem, 
an expression for the average value of the signal-to-interference-noise ratio (SINR) at the D2D receiver is 
derived in closed-form. 
Our theoretical derivations are corroborated through simulation results, and various system design 
insights are deduced. It is indicatively showcased that the proposed RIS-empowered underlaid D2D system design 
outperforms the benchmark semi-adaptive optimal power and optimal distance schemes, offering $44\%$ and $20\%$ 
performance improvement, respectively.
\end{abstract}

\section{Introduction}
\label{sec_intro} 
Sixth-generation (6G) networks are envisaged to cater to a large number of connected user terminals, and hence, there is a 
huge demand for radio frequency spectrum as well as high data rate traffic \cite{Akyildiz_6G_2020}. 
The ever-increasing demand for spectrum can be met through device-to-device (D2D) communications. This paradigm enables direct wireless 
connections of battery-limited devices \cite{D2D_survey}. Specifically, a D2D user pair communicates through reusing the 
uplink spectrum dedicated to cellular users, and thus, confronts the spectrum scarcity issue. Additionally, as the 
D2D users are physically adjacent to each other, they can convey messages without the involvement of a base station (BS). 
Thus, the communication overhead induced by the latter implementation is significantly reduced \cite{Zhang_Jo_d2d_EE_swipt}. 
Overlay and underlay are the two methods of spectrum sharing. In the underlay method, the BS and the D2D transmitter 
operate simultaneously on the same frequency range, maintaining the induced interference power at the cellular receivers 
below a specified threshold. However, the interfering signals coming from a BS and the inherent randomness of a wireless channel limit the 
quality of the signal received at the D2D user. 
In the overlay method, the D2D pair interacts through a dedicated channel, hence, both the D2D user pair and the 
cellular users are free from ensuing interference. 

The design and optimization of underlaid cognitive radio systems, both standalone and in conjunction with emerging beyond 
fifth-generation (5G) technologies, has been the subject of extensive research and development endeavors \cite{D2D_survey_2022}. Indicatively, \cite{George_CR_powalloc_harvst_tx, George_green_CR} studied underlaid cognitive radios with energy harvesting-enabled secondary transmissions, while \cite{George_spectrum_sensing} investigated the issue of data detection and spectrum sensing performed by the secondary receiver. 
Performance analyses of relay selection under cognitive radio setting were explored in \cite{George_CR_AF_relsel, George_BER_underlay_relsel, George_OP_CR_relsel}. 
Recently, there has been a lot of interest in reconfigurable intelligent surfaces (RISs) 
as a comprehensive technology for 6G communication networks \cite{George_RIS_survey_ch_model}. 
Particularly, RISs are metasurfaces composed of a collection of 
ultra-low power, programmable, and possibly sub-wavelength spaced elements of tunable 
electromagnetic responses that have the ability to alter the phase of impinging radio waves dynamically
\cite{George_RIS_EE, Renzo_George_smart_meta_surface,Larsson_RIS_chnl_mdl, WavePropTCCN, George_RIS_6G_connect}. Thus, they can intelligently configure the signal propagation environment to enhance 
the signal towards intended users, and/ or weaken the signal by creating interference through precise beam steering \cite{George_RIS_survey_expt}. 
As a result, the RIS technology is viewed~\cite{RIS_standards} as a cost-efficient and energy-saving solution to address the coverage problem, facilitate the 
high capacity requirements of end-user, and enable various localization and sensing applications, as well as integrated sensing and communications \cite{George_RIS_DF_hybrid,George_RIS_deployment,Keykhosravi2022infeasible,RIS_ISAC_2023}. 
In \cite{Han_orientation_locn_opt}, the authors studied the optimal RIS placement that maximizes the coverage area of an 
RIS-aided downlink system. The placement of multiple RISs for maximizing the signal-to-noise (SNR) for a transmitter-receiver 
communication pair was analyzed in~\cite{moustakas_alexandg_2023}, unveiling the regimes where the RIS need to be placed, 
namely, either to the midpoint between the nodes or close to any of them. The latter analysis complemented the relevant 
literature~\cite{Larsson_RIS_chnl_mdl, Ntontin_RIS_placement, amplifying_RIS_2022}. Following the recent trend of computationally autonomous 
RISs for sensing and self-configurability~\cite{hardware2020icassp,HRIS_Mag}, a target sensing framework at the RIS side 
was developed in~\cite{Zhang_targetsense_IRS}. In that work, the RIS reflection coefficients and the angle of reflection 
from the target that maximize the received signal power at the sensors were obtained. 

\subsection{State-of-the Art}
\label{sec_rvw} 
For their ability to create smart radio environments via dynamic reflective beamforming and over-the-air analog computing \cite{George_RIS_survey_expt, Space_shift_keying_RIS, PhysFad}, 
RISs have shown a substantial potential to solve various design issues coupled with traditional cognitive radios.  
In several communication scenarios, RISs have been considered an effective means to mitigate interference in 
cognitive radio networks \cite{Guan_ruizhang_IRS_CRN, Yang_outg_cog_RIS}, and cellular Internet-of-Things (IoT) 
\cite{Yu_Zhang_IRS_cellular_iot}. It has been actually proved that RISs can enhance the quality-of-service through phase 
alignment, assuming complete knowledge of the channel state, as well as localization information of the 
user \cite{Kudathanthirige_Amarasuriya_IRS_rayleigh, George_RIS_overhead}. Under such phase control settings, 
few recent works have carried out outage probability (OP) analyses of D2D frameworks. In particular, D2D networks in the 
underlay mode without a direct path between the D2D pairs were investigated in 
\cite{Ni_Zhu_IRS_D2D_nakagami_perfana, Yang_outg_cog_RIS, khoshafa_telex_IRS_PLS_D2D, Kaleem_IRS_FD_jam_D2D}. 
In the presence of a cellular user, closed-form expressions for the OP at the D2D user were presented in
\cite{Ni_Zhu_IRS_D2D_nakagami_perfana, khoshafa_telex_IRS_PLS_D2D} under different settings. 
Using Nakagami-$m$ channels, the performance of an RIS-aided framework was investigated in 
\cite{George_RIS_perfana_nakagami} (note that $m=1$ corresponds to the Rayleigh distribution). 
In \cite{khoshafa_telex_IRS_PLS_D2D}, a similar underlaid D2D framework with a single eavesdropper was considered, and the secrecy OP of the system was studied. 
A similar system including a full-duplex receiver with jamming capability was studied in \cite{Kaleem_IRS_FD_jam_D2D}. 
In \cite{George_RIS_secrecy_AN, MaliciousRIS_2022}, the maximization of the secrecy rate was studied for different cases of 
channel state information (CSI) availability, where both the legitimate user and eavesdropper were assumed to be assisted by respective RISs.

Very recently, considerable attention has been paid to the optimization of RIS-aided spectrum-sharing communication 
networks involving cognitive radios, IoT, and D2D. To solve design problems related to RIS-assisted systems, several authors 
have attempted to optimize some of the key performance metrics. 
In an RIS-assisted cognitive radio network setup, aiming to jointly design the transmit precoder (TP) at the 
secondary receiver (SR) and the phase profile at the RIS, 
the authors maximized the rate of the SR \cite{Larsson_cog_IRS_opt, Schober_IRS_MIMO_cog}. The work in 
\cite{Cunhua_IRS_MIMO_cog} emphasised on the optimization of a weighted sum of the achievable rates. 
Those optimization problem formulations considered constraints on the transmit power at the secondary transmitter (ST) and the
interference power of the primary receivers (PR). 
An SNR constraint on SRs and a unit-modulus constraint on RIS phase response 
were considered in \cite{Cunhua_IRS_cog_impCSI}, where the goal was the minimization of the ST transmit power. 
The authors in \cite{Qingqing_IRS_cog_SE} investigated the spectral efficiency (SE) maximization problem and 
obtained the optimal TP and the angle of elevation of the ST along with the phase profile of the RIS.
In \cite{Gong_hanzo_BF_SWIPT_IoT}, the authors introduced an RIS-assisted energy harvesting IoT framework 
and derived the optimized TP and discrete phase shift of the RIS that maximize the least signal-to-interference-noise 
ratio (SINR). In another work\cite{HMWang_IRS_cog_SR_max}, the secrecy rate maximization was investigated.

The interference temperature and the computational capabilities of devices are two of the main shortcomings of 
D2D communication systems. The adoption of an RIS can actually contribute in overcoming the first issue through robust 
interference mitigation design, which can be offered via the RIS reflective beamforming. 
This reflective beamforming can be rather directive in low angular spread channels. 
The authors in \cite{Chen_Poor_IRS_d2d} obtained the  joint optimal phase profile at the RIS and the power allocated to D2D sources 
that maximize the sum-rate performance. 
The joint maximization of the energy efficiency (EE) and SE was investigated in 
\cite{Yang_IRS_cellular}, where the optimal values of the transmit power, RIS phase profile, spectrum reuse factor, 
and the TP at the BS were derived. The authors in \cite{jose_anirudh_outage_min} optimized the power allocated to each device cluster 
with the goal of minimizing the OP of non-orthogonal multiple access underlaid D2D networks using meta-heuristics approaches. In addition, several researchers have considered RISs in D2D networks to optimize various other performance metrics,  
e.g., the D2D communication rate \cite{Xu_sumrate_opt_IRS_D2D, Xie_weightsumrate_opt_IRS_D2D}, SINR at D2D user \cite{ICASSP23_RIS_SINR}\footnote{The conference version of this work considered a similar underlaid RIS-empowered D2D communication system, but formulated and solved a different design problem having the SINR at the D2D receiver as the optimization objective.}, 
BS transmit power \cite{Sugiura_IRS_secure_SEE_max}, EE \cite{Rose_IRS_EE}, and secure EE \cite{Zeng_IRS_cog_SEE_max}. 
It is also noted that RISs are lately appearing in several other concepts with the focus being their phase profile optimization. Indicatively, an RIS-aided system with multiple unmanned aerial vehicles (UAVs) was studied in \cite{Duong2022_RIS_UAV} aiming to enhance the signal quality at the receiver. An efficient technique based on deep reinforcement learning~\cite{pervasive_DRL_RIS} was devised to jointly optimize the transmit power at the UAVs and the RIS phase shifts. In~\cite{Poor2020_improper}, it was concluded that improper Gaussian signaling offers increased data rates, as compared to a typical proper Gaussian signaling method. In \cite{Q_Wu_JSAC_2020}, simultaneous wireless information and power transfer (SWIPT) with multiple RISs was explored, where it was found that SWIPT reduces the total transmit power budget, thus, helping RISs to achieve better energy beamforming gain.

\subsection{Motivation and Contributions}
To the best of the authors' knowledge, the open technical literature on RIS-assisted D2D communication systems does not include any 
OP analysis for the general case where there exist direct links between the D2D user pairs. It is noted that 
D2D communications usually take place over short distances, hence, it is highly probable that direct links are 
present~\cite{Yang_IRS_cellular, Yang_outg_cog_RIS}. In such practical cases, the extra RIS-induced reflected path will 
contribute to a diversity gain. In addition, the available literature lacks designs for the optimal OP-minimizing resource allocation between the D2D source and the RIS placement, under cognitive radio constraints. In this paper, we address this research gap, by designing a spectrally efficient underlaid D2D system with ultra-low-power IoT devices which incorporates an inexpensive and energy-efficient reflective RIS.
The key contributions of this paper are summarized as follows. 
\begin{itemize}
\item We consider an RIS-assisted cognitive communication system that includes the direct link between the D2D 
communication pair. We adopt a phase-shift control mechanism at the RIS reflecting elements such that all the 
reflected channels are combined coherently with the direct communication channel. 
\item We present a novel analysis for the system's underlying OP, which is based on the central limit theorem. Additionally, we provide a statistical characterization of the SINR metric at the D2D user 
via an integral-form expression that includes the Gaussian $Q$-function.
To obtain analytical insights on the latter metric, we approximate the $Q$-function using an exponential series, 
which is invoked to define two different SNR regimes. 
\item We formulate a novel optimization framework that focuses on the OP minimization of the proposed RIS-empowered underlaid D2D communication system. For this design objective, 
we derive the jointly optimal RIS placement and power allocated to the D2D transmitter. The proposed joint optimization takes into account 
both a parallel and an elliptical topologies for the RIS deployment. 
Since the objective function involves the OP performance which is expressed in an integral form, we present an equivalent tractable transformation, whose globally optimal solution is derived.
\item A thorough numerical investigation is presented to validate the performance analysis, the 
quality of the considered approximations, and the global optimality of the proposed joint design. 
This investigation also offers valuable engineering insights and quantifies the achievable performance gains.
It is showcased that the devised system configuration strategy outperforms the benchmark optimal power allocation method by $44\%$. 
\end{itemize}
\subsection{Organization and Notation}\label{sec_org}
The remainder of this paper is structured as follows. 
Section~\ref{sec_sys_model} introduces the system and channel models. 
In Section~\ref{sec_perf_ana}, the SINR statistics and an investigation of the OP performance with and 
without a direct link are presented. The distribution of the SNR of the BS to the D2D receiver channel and 
approximations for the average SINR at the D2D receiver, along with the proposed framework for optimization, are included in Section~\ref{sec_joint_opt}. 
This section also outlines the proposed equivalent problem formulation to the considered design objective for two different  network topologies. In Section~\ref{sec_opt_sol}, we present the optimal solutions for RIS deployment and power allocation for the D2D source. 
Numerical results and insights are discussed in Section~\ref{sec_results}, while  Section~\ref{sec_conclude} contains the paper's concluding remarks.

Italicized letters and bold-face lower-case letters denote scalars and matrices, respectively.
$\mathbb{C}^{M\times N}$ stands for the $M\times N$ complex matrix. Also, 
$\mathbf{x}^H$ and $\norm{\textbf{x}}$ stand for the conjugate-transpose and $l_2$ norm of $\textbf{x}$, 
respectively. $\textrm{diag}(\textbf{x})$ creates a diagonal matrix with the elements of $\textbf{x}$ 
placed along the principal diagonal. 
$\abs{x}$ and $\angle{x}$ are used to denote the magnitude and phase of $x$, respectively. 
The circularly symmetric complex Gaussian (CSCG) is represented by the notation $\mathcal{CN}(\mu,\sigma^2)$, 
where $\mu$ and $\sigma^2$ denote the mean and variance, respectively.
The expectation and variance operations are denoted by $\mathbb{E}[\cdot]$ and $\mathbb{V}[\cdot]$, respectively, whereas $\mathbb{P}[\cdot]$ returns the probability. Finally, $\overline{F_{A}}(\cdot)=1-F_{A}(\cdot)$ 
denotes the complement of the cumulative distribution function (CCDF) of the random variable $A$, and $A\circ B(\cdot)$ 
stands for the composite function of functions $A(\cdot)$ and $B(\cdot)$. Table \ref{tab_sym} includes various symbols used 
along with their meaning.

\section{System and Channel Models}
\label{sec_sys_model}
\begin{table}
\footnotesize
\caption{Table of symbols}
\begin{tabular}{|p{1.1cm}| p{7cm} | } 
\hline 
{Symbols} & {Description}\\
\hline
$P_{\rm b}$ & Transmit power budget of BS \\ \hline
$P_{\rm s}$ & Transmit power budget of DS \\ \hline
$y_{\rm j}$ & Received signal at node j \\ \hline
$n_{\rm j}$ & AWGN at node j \\ \hline
$N_0$ & Variance of AWGN \\ \hline 
$M$ & Number of antennas at the BS \\ \hline 
$N$ & Number of reflecting elements at the RIS \\ \hline 
$\alpha_n$ & Amplitude of $n$-th element at RIS \\ \hline
$\theta_n$ & Phase shift of $n$-th element at RIS \\ \hline
$h_{\rm ij}$ & Channel coefficient of i- j link \\ \hline 
$\beta_{\rm ij}$ & Mean of channel gain of i- j link \\ \hline 
$s_{\rm b}$ & Transmit signal at BS \\ \hline
$s_{\rm s}$ & Transmit signal at DS \\ \hline 
$\bar{\gamma}_{\rm j}$ & Average transmit SNR at node j \\ \hline
$\Gamma_{\rm d}$ & SINR at DU \\ \hline 
$I_{\rm b}$ & Interference at CU due to D2D transmission \\ \hline 
$I^{\rm th}$ & Interference threshold \\ \hline 
$\gamma_{\rm th}$ & SINR threshold for D2D transmission \\ \hline 
$p_{\rm out}$ & Outage probability of the D2D link \\ \hline
\end{tabular} \label{tab_sym}
\end{table}
\subsection{Network Topology and Access Protocol}
\label{subsec_topology}
Let us consider an RIS-aided underlaid D2D communication system, as illustrated in Fig. \ref{FIG_sys_model}. It consists of a 
single cellular BS, a single RIS, and a single D2D communication pair (one D2D transmitter (DS) and one D2D receiver (DU)), as well as a cellular user (CU). 
In particular, the D2D pair is allowed to utilize the same spectrum with the cellular network, hence, its communication creates interference when realized simultaneously with cellular transmissions. 
Unlike \cite{khoshafa_telex_IRS_PLS_D2D, Ni_Zhu_IRS_D2D_nakagami_perfana}, we assume that DU is directly reachable from DS. In addition, this communication is assumed to be boosted by an RIS~\cite{George_RIS_survey_expt}. 
Clearly, these RIS-boosted D2D transmissions interfere cellular communications. Likewise, DU experiences interference signals from the cellular transmissions. 

We make the assumption that the RIS to CU interference is negligible; this can be accomplished through carefully optimized RIS placement and area of influence~\cite{EURASIP_RIS} aiming to only impact the small geographic area of the intended D2D communications. 
In addition, the deployed reflective RIS is devoid of any active radio-frequency chains, hence, it merely performs phase adjustments of its reflecting elements, which consumes a negligible amount of power~\cite{George_RIS_survey_ch_model}.
We therefore take into account the signal that is reflected by the RIS only once, and neglect those signals which undergo twice or multiple times reflections~\cite{RIS_tacit_2023}. 
\begin{figure}[!t]
\centering
\includegraphics[width=\columnwidth] {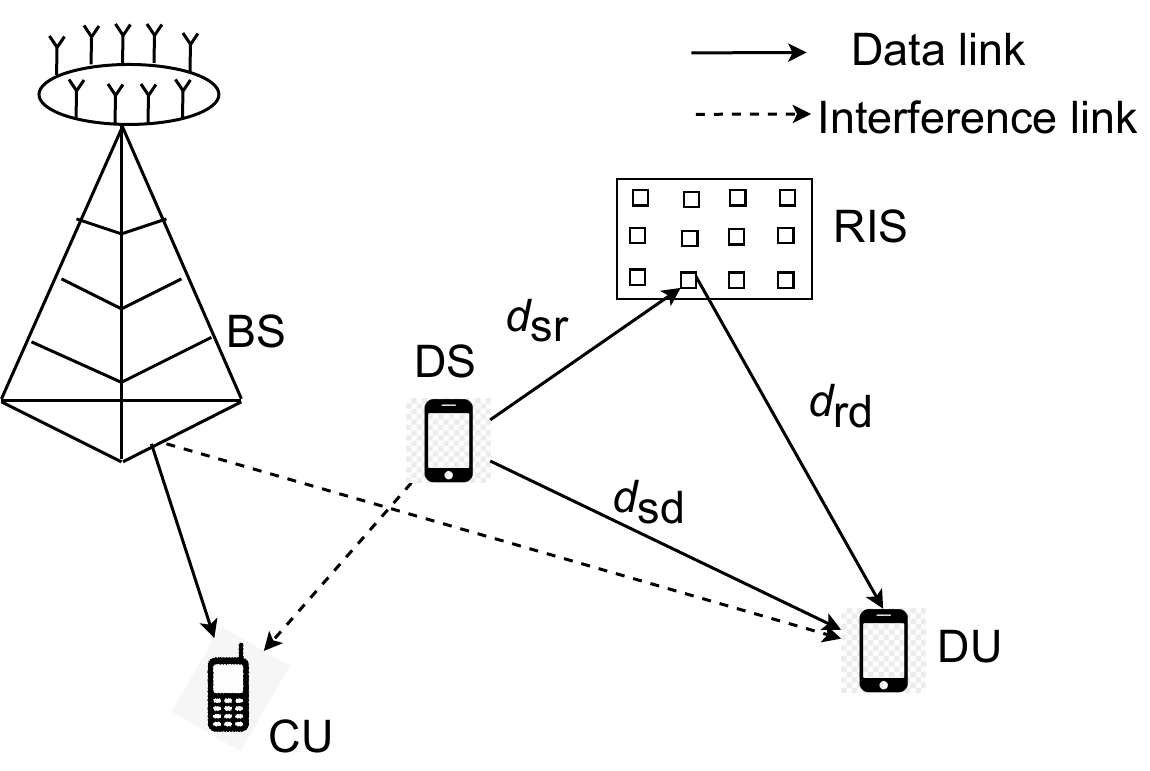}
\caption{The RIS-empowered underlay D2D communication system model under consideration.}
\label{FIG_sys_model}
\end{figure} 
Finally, the RIS is optimally placed according to the following two topologies (Fig.~\ref{FIG_topology}): 
\subsubsection{Parallel} Assuming DS is placed at point ($x_0$, $y_0$) and DU is at point ($x_0+d_{\rm sd}$, $y_0$), the RIS is 
placed at $(x_0+d,y_0+y)$ which is on a path parallel to the DS-DU line segment. 
Hence, it holds $d_{\rm sr}=\sqrt{d^2+y^2}$ and $d_{\rm rd}=\sqrt{(d_{\rm sd}-d)^2+y^2}$. 

\begin{figure}[!t]
\centering
\includegraphics[width=\columnwidth] {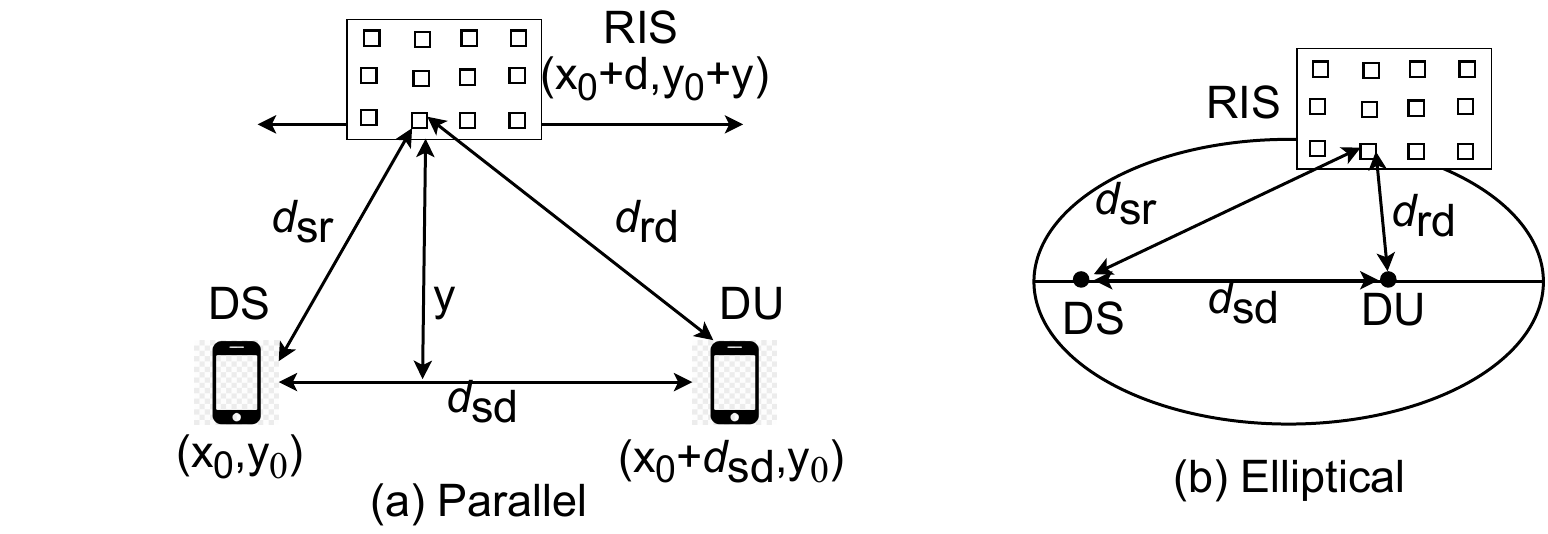}
\caption{The topology of the simulated RIS-aided underlaid D2D communication system.}
\label{FIG_topology}
\end{figure} 
\subsubsection{Elliptical} The RIS is placed along the locus of an elliptical path, where the DS and DU are located 
on the two foci of the ellipse with a separation equal to $d_{\rm sd}$. Therefore, it yields that $d_{\rm sr}=d$ and $d_{\rm rd}=d_{\rm sd}/\epsilon - d$, where $\epsilon$ represents the eccentricity of the ellipse \cite{DM_ICC_2016}. 

\subsection{RIS and Channel Models}
\label{subsec_RIS_chnl}
We assume that DS, DU, and CU are single-antenna nodes, while the BS is equipped with an $M$-antenna array and the RIS comprises $N$ reflecting elements. The fading channel coefficients between DS and RIS, RIS and DU, DS and DU, DS and CU, BS and CU, and between 
BS and DU are denoted by $\mathbf{h}_{\rm sr} \in \mathbb{C}^{N\times 1}$, $\mathbf{h}_{\rm rd} \in \mathbb{C}^{N\times 1}$, $h_{\rm sd}\in \mathbb{C}^{1}$, 
$h_{\rm sc}\in \mathbb{C}^{1}$, $\mathbf{h}_{\rm bc} \in \mathbb{C}^{M\times 1}$, $\mathbf{h}_{\rm bd}\in \mathbb{C}^{M\times 1}$, respectively. 
Specifically, $\mathbf{h}_{\rm sr}\triangleq \lb[h_{\text{sr},1},\ldots,h_{\text{sr},n},\ldots,h_{\text{sr},N}\rb]^T$ and 
$\mathbf{h}_{\rm rd}\triangleq\lb[h_{\text{rd},1},\ldots,h_{\text{rd},n},\ldots,h_{\text{rd},N}\rb]^T$. 
We assume that the receivers and the RIS possess perfect knowledge of the instantaneous CSI of the channels they are 
involved at every fading duration. We also make the reasonable assumption that the antennas at D2D users and RIS elements are placed adequately apart, hence, the respective channel gains between different pairs of antennas undergo independent fading. 
In particular, the RIS is assumed to be placed in the Fraunhofer region of DS and DU 
\cite{Zhang_IRS_farfield,Dai_IRS_farfield}, hence, the distances between DS and each $n^{\textrm{th}}$ ($n=1,2,\ldots,N$) RIS element are the same; the same holds for the distances between each $n^{\textrm{th}}$ element of the RIS and DU.

The channel coefficients are considered mutually independent and follow the Rayleigh distribution. 
Particularly, the channel between each i-j link, with $\rm ij \in \{\rm sr, rd, sd, sc, bc, bd\}$, is given 
by \cite{Kudathanthirige_Amarasuriya_IRS_rayleigh}:
\begin{align}
\label{eq_chnl_path_loss}
h_{\rm ij}&\triangleq\sqrt{\beta_{\rm ij}}\tilde{h}_{\rm ij},
\end{align}
where $\beta_{\rm ij}$ represents the path loss \cite{Ellingson_pathloss} of the i-j link and 
$\tilde{h}_{\rm ij}\sim \mathcal{CN}(0,1)$.
The channel gains, denoted as $\abs{h_{\rm ij}}^2$, are exponentially distributed random variables having the mean: 
\begin{align}
\label{eq_PL}
\beta_{\rm ij}&=C\mathbb{E}\lb[\abs{h_{\rm ij}}^2\rb]=\lb(\frac{d_0}{d_{\rm ij}}\rb)^{\eta}, 
\end{align}
where $d_0$ and $\eta$ stand for the reference distance and the path loss exponent, respectively.  
$C$ is the path loss at the reference point $d_0$.
For the considered parallel topology, the parameters $\beta_{\rm sr}$ and $\beta_{\rm rd}$ are obtained as: 
\begin{align}
\label{eq_beta_parallel}
\beta_{\rm sr}=\lb(\frac{d_0}{\sqrt{d^2+y^2}}\rb)^\eta,~ 
\beta_{\rm rd}=\lb(\frac{d_0}{\sqrt{(d_{\rm sd}-d)^2+y^2}}\rb)^\eta,
\end{align} 
while, for the elliptical topology, as follows: 
\begin{align}
\label{eq_beta_ellipse}
\beta_{\rm sr}=\lb(\frac{d_0}{d}\rb)^\eta, ~\beta_{\rm rd}=\lb(\frac{d_0}{d_{\rm sd}/\epsilon-d}\rb)^\eta. 
\end{align}
The transmitted signal from BS to CU is designed as $\mathbf{x}_{\rm b} = \mathbf{w}_{\rm b}s_{\rm b} \in \mathbb{C}^{M\times 1}$, 
where $\mathbf{w}_{\rm b} \in \mathbb{C}^{M\times 1}$ is the beamforming vector. 
Simultaneously, DS transmits the signal $s_{\rm s}$ to DU. Without loss of generality, we assume that $s_{\rm b}$, $s_{\rm s}\sim\mathcal{CN}(0,1)$. 
Hence, the power budgets at DS and BS are represented by $P_{\rm s}$ and $P_{\rm b}=\norm{\mathbf{w}_{\rm b}}^2$, 
respectively. 
Furthermore, we respectively express the links between each $n^{\textrm{th}}$ RIS element and DS as well as DU as 
\cite{Kudathanthirige_Amarasuriya_IRS_rayleigh}
\begin{align}
h_{\text{sr},n}=\abs{h_{\text{sr},n} } e^{j\angle h_{\text{sr},n}} , 
~~h_{\text{rd},n}=\lvert h_{\text{rd},n} \rvert e^{j\angle h_{\text{rd},n}},
\end{align}	
where both $\angle h_{\text{sr},n}$ and $\angle h_{\text{rd},n}$ follow the uniform distribution within the 
range $(-\pi, \pi]$. 
The response of the RIS is captured in the phase-shift matrix, which is defined as
\begin{align}
\label{eq_diag_phase}
\mathbf{\Theta}\triangleq \diag\lb(\alpha_1 e^{j\theta_1}, \cdots, \alpha_n e^{j\theta_n},\cdots,
\alpha_N e^{j\theta_N}\rb),
\end{align}
where $\alpha_n\in \lb[0,1\rb]$ and $\theta_n\in \lb[0,2\pi\rb)$ stand for the amplitude and the phase shift of each $n^{\textrm{th}}$ element of the RIS, respectively.

\subsection{Received Signal and SINR Modeling}
\label{subsec_sig_analysis}
The signals received at DU and CU can be respectively expressed as follows:
\begin{align}
\label{eq_rx_DU}
y_{\rm s}&\triangleq\sqrt{P_{\rm s}}\lb(h_{\rm sd}+\mathbf{h}^H_{\rm rd}\mathbf{\Theta} \mathbf{h}_{\rm sr}\rb)s_{\rm s} + 
\mathbf{w}^H_{\rm b}\mathbf{h}_{\rm bd}s_{\rm b} +n_{\rm s},  \\
\label{eq_rx_CU}
y_{\rm c}&\triangleq\mathbf{w}^H_{\rm b}\mathbf{h}_{\rm bc}s_{\rm b}+\sqrt{P_{\rm s}}h_{\rm sc} s_{\rm s}+n_{\rm c},
\end{align}
with $n_{\rm s}$ and $n_{\rm c}$ representing the additive white Gaussain noise terms at DU and CU, respectively. 
These terms are distributed as $n_{\rm i}\sim \mathcal{CN}(0, N_0)$ with $\rm i\in \{d, c\}$. 

Following the assumption that the knowledge of $\abs{h_{\text{sr},n}}$, $\abs{h_{\text{rd},n}}$ and $\mathbf{\Theta}$ is available at DU, this node treats 
the rest of the received signals as interference. From expression \eqref{eq_rx_DU}, the SINR at DU can be expressed as: 
\begin{align}
\label{eq_snr_DU_beta_theta}
\Gamma_{\rm d}&\triangleq\frac{P_{\rm s}\abs{h_{\rm sd}+\sum\limits^N_{n=1}\alpha_n\abs{h_{\text{sr},n}} \abs{h_{\text{rd},n}} 
e^{j\lb(\theta_n+\angle h_{\text{sr},n}+\angle h_{\text{rd},n}\rb)}}^2}
{\abs{\mathbf{w}^{H}_{\rm b}\mathbf{h}_{\rm bd}}^2+N_0}.
\end{align}To obtain the maximum SINR at DU, 
the phases of the $N$ paths of the cascaded DS-RIS-DU channel can be intelligently aligned with the phase of the direct channel \cite{Larsson_RIS_chnl_mdl}. 
This phase configuration is given for each $n^{\textrm{th}}$ RIS element by \cite{Kudathanthirige_Amarasuriya_IRS_rayleigh,khoshafa_telex_IRS_PLS_D2D} 
\begin{align}
\label{eq_phase_align}
\theta^{\star}_n&=\angle h_{\rm sd}-\lb(\angle h_{\text{sr},n}+\angle h_{\text{rd},n}\rb). 
\end{align}

Employing digital beamforming at multiantenna BS would require several radio frequency (RF) chains equal to the number 
of antennas involved. As these RF chains comprise of several active components, namely, amplifier, converter, mixer, 
filter and others, implementing beamforming at BS will increase the hardware complexity and cost \cite{Mehta_TAS}. 
Therefore, at the BS, the transmit antenna offering the maximum instantaneous SNR is chosen as the best transmit 
antenna \cite{khoshafa_telex_IRS_PLS_D2D, Zeng_IRS_cog_SEE_max}. This transmit antenna selection (scheme) requires a 
single RF chain at the multiantenna BS. Hence, the effective SNR of the BS-DU channel is obtained as:

\begin{align}
\label{eq_BS_ant_sel_bd}
\gamma_{\rm bd}\triangleq\max_{m \in \{1,2,\ldots, M\}}\gamma_{\text{bd},m}, 
\end{align}
where $\gamma_{\text{bd}, m}=\bar{\gamma}_{\rm b}\abs{h_{\text{bd},m}}^2$ and $\bar{\gamma}_{\rm b}=P_{\rm b}/N_0$.
By substituting \eqref{eq_phase_align} into \eqref{eq_snr_DU_beta_theta}, the SINR at DU can be expressed as follows:
\begin{align}
\label{eq_snr_DU_fin}
\Gamma_{\rm d}&=\frac{\bar{\gamma}_{\rm s}\abs{Y}^2}{\gamma_{\rm bd}+1}=\frac{\gamma_{\rm srd}}{\gamma_{\rm v}},
\end{align}
where 
$\gamma_{\rm srd}\triangleq \bar{\gamma}_{\rm s}\abs{Y}^2$,  
$Y\triangleq h_{\rm sd}+X$, $X\triangleq\sum^N_{n=1} \alpha_n\abs{h_{\text{sr},n}} \abs{h_{\text{rd},n}}$, 
$\bar{\gamma}_{\rm s}\triangleq P_{\rm s}/N_0$, $\gamma_{\rm v}\triangleq \gamma_{\rm bd}+1$, and 
$\gamma_{\rm bd}\triangleq\bar{\gamma}_{\rm b}\max_m\abs{h_{\text{bd},m}}^2$. 

Due to the simultaneous transmissions of the DS over the same frequency band as BS, the induced interference temperature $I_{\rm b}$ is required to be maintained below an interference threshold $I^{\textrm{th}}$ that can tolerated by the CU. In mathematical terms, it should hold for the D2D-induced interference:  
\begin{align}
\label{eq_interference_powr_lim}
I_{\rm b}\triangleq P_{\rm s}\mathbb{E}\lb[\lvert h_{\rm sc}\rvert^2\rb] \le I^{\textrm{th}}.
\end{align}

\section{Outage Probability Analysis} 
\label{sec_perf_ana}
In this section, we analyze the OP performance of the proposed RIS-aided cognitive communication system. We first obtain the distribution of the SINRs in subsection~ \ref{sec_dist_sinr}. Subsection~\ref{subsec_OP}
 includes the OP analysis when the direct link between the D2D communication pair is present, while subsection~\ref{sec_without_dir_ana} deals with the case where this link is absent. 

\subsection{Distribution of the SINRs} 
\label{sec_dist_sinr}
Starting from the distribution of the DS-to-RIS and RIS-to-DU channels, we obtain the distribution of the sum of $N$ such 
reflected paths; this is represented by the random variable $X$. Next, we obtain the distribution of the sum of the DS-to-DU direct channel and DS-RIS-DU reflected channel; this is represented by $Y\triangleq h_{\rm sd}+X$. Finally, the distribution of the squared SNR, i.e., $Y^2$, is obtained. Similarly, we derive the distribution of the BS-to-DU channel. 

Since the random variables $\abs{h_{\text{sr},n}}$ and $\abs{h_{\text{rd},n}}$ follow the Rayleigh distribution, 
$X$ is defined as the sum of the product of two Rayleigh random variables. 
When the number $N$ of the RIS reflecting elements becomes large enough, $X$ converges to a Gaussian distribution. 
This results in a tractable statistical characterization. 
Assuming that the channel realizations are independent and identically distributed, 
i.e., $\alpha_n=\alpha$ $\forall n$, it becomes $X\sim\mathcal{CN}(\mu,\sigma^2)$ with the mean and variance given respectively by\cite{Tsiftsis_twoway_perfana_opt}:
\begin{equation}
 \begin{aligned}
\label{eq_mean_var_gaussian}
 \mu&=\frac{N\alpha\pi}{4}\sqrt{\beta_{\rm sr}\beta_{\rm rd}},
\\
 \sigma^2&= N\alpha^2\beta_{\rm sr}\beta_{\rm rd}\lb(1-\frac{\pi^2}{16}\rb).
\end{aligned}
\end{equation}
Therefore, from \cite{Kudathanthirige_Amarasuriya_IRS_rayleigh}, 
we can derive the cumulative distribution function (CDF) of $X$ as follows:
\begin{align}
\label{eq_pdf_cdf_gaussian}
 F_X(x)&=1-\omega Q\lb(\frac{x-\mu}{\sigma}\rb),
\end{align}
where
$Q(\cdot)$ denotes the $Q$-function defined in \cite{book_papoulis_pillai} with $\omega\triangleq1/Q(-\mu/\sigma)$. 

The distribution of $Y$ can be derived as ($a_t\triangleq\frac{1}{\beta_{\rm sd}}+\frac{1}{2\sigma^2}$):
\begin{align}
\label{eq_cdf_Y_conv}
&F_{Y}(y)=\mathbb{P}\lb[h_{\rm sd}+X \le y\rb] \nn \\
&=\int^y_0 F_{h_{\rm sd}}(y-x) f_X(x)dx \nn \\
&=1-\omega Q\lb(\frac{y-\mu}{\sigma}\rb) 
-\frac{\omega}{\sqrt{2a_t\sigma^2}}\exp\lb(-\frac{\lb(y-\mu\rb)^2}{2\sigma^2\beta_{\rm sd}a_t}\rb) \nn \\
&\times\lb[1-Q\lb(\sqrt{\frac{2}{a_t}}\lb(\frac{y}{\beta_{\rm sd}}+\frac{\mu}{2\sigma^2}\rb)\rb) 
-Q\lb(\sqrt{\frac{2}{a_t}}\lb(\frac{y-\mu}{2\sigma^2}\rb)\rb)\rb].
\end{align}
This result can be used for the derivation of the CDF of $\gamma_{\rm srd}=\bar{\gamma}_{\rm s}\abs{Y}^2$, using a standard transformation of random variables. In particular, it holds for this CDF that $F_{\gamma_{\rm srd}}(x)=F_{Y}(\sqrt{x/\bar{\gamma}_{\rm s}})$, which, using \eqref{eq_cdf_Y_conv}, yields the expression: 
\begin{align}
\label{eq_cdf_conv_square}
&F_{\gamma_{\rm srd}}(x)
=1-\omega Q\lb(\frac{\sqrt{x/\bar{\gamma}_{\rm s}}-\mu}{\sigma}\rb) 
-\frac{\omega}{\sqrt{2a_t\sigma^2}} \nn \\
&\times\!\exp\lb(\!-\frac{\lb(\sqrt{x/\bar{\gamma}_{\rm s}}-\mu\rb)^2}{2\sigma^2\beta_{\rm sd}a_t}\!\rb)\!\!
\lb[1\!-\!Q\lb(\sqrt{\frac{2}{a_t}}\lb(\frac{\sqrt{x/\bar{\gamma}_{\rm s}}}{\beta_{\rm sd}}+\frac{\mu}{2\sigma^2}\rb)\!\rb) \rb.\nn \\
&\lb. -Q\lb(\sqrt{\frac{2}{a_t}}\lb(\frac{\sqrt{x/\bar{\gamma}_{\rm s}}-\mu}{2\sigma^2}\rb)\rb)\rb].
\end{align}
\begin{remark}[Absence of the Direct Link]
When the distance between DS and DU is large, i.e., when $\beta_{\rm sd}\rightarrow 0$, it holds for the third term of 
\eqref{eq_cdf_conv_square} that 
$\exp\lb(-\frac{\lb(\sqrt{x/\bar{\gamma}_{\rm s}}-\mu\rb)^2}{2\sigma^2\beta_{\rm sd}a_t}\rb)\rightarrow0$, and that 
$Q\lb(\sqrt{\frac{2}{a_t}}\lb(\frac{\sqrt{x/\bar{\gamma}_{\rm s}}}{\beta_{\rm sd}}+\frac{\mu}{2\sigma^2}\rb)\rb)$, 
$Q\lb(\sqrt{\frac{2}{a_t}}\lb(\frac{\sqrt{x/\bar{\gamma}_{\rm s}}-\mu}{2\sigma^2}\rb)\rb)\rightarrow 1/2$. 
Hence, $F_{\gamma_{\rm srd}}(x)\rightarrow 1-\omega Q\lb(\frac{\sqrt{x/\bar{\gamma}_{\rm s}}-\mu}{\sigma}\rb)$. For this case, 
\eqref{eq_cdf_conv_square} reduces to the distribution of $\Gamma_{\rm d}$ when there is no direct link between DS and DU.  
\end{remark}

We next intend to use $F_{\gamma_{\rm srd}}(\cdot)$ to analytically obtain the OP performance, $p_{\rm {out}}$. However, the $Q$-function in \eqref{eq_cdf_conv_square} does not have a closed-form. To this end, we focus on the high SNR regime and deploy a sufficiently tight approximation for this function. By using \cite{Q_as_sum_exp}, the $Q$-function can be approximated using a sum of exponential functions as follows: 
\begin{subequations}\label{eq_sum_Q_func}
\begin{align}
\label{eq_sum_Q_func_more}
Q(x)&\approx
\sum^4_{k=1} c_k e^{-p_k x^2}, ~~~ x \ge 0,\\
\label{eq_sum_Q_func_less}
&\approx1-\sum^4_{k=1} c_k e^{-p_k x^2},~~~ x < 0,
\end{align}
\end{subequations}
where $c=\{\frac{1}{16}, \frac{1}{8}, \frac{1}{8}, \frac{1}{8}\}$, $p=\{\frac{1}{2}, 1, \frac{10}{3}, \frac{10}{17}\}$. 
With this approximation, the PDF of $\gamma_{\rm v}$ in \eqref{eq_snr_DU_fin} can be expressed as: 
\begin{align}
\label{eq_pdf_BS_ant_sel_bd}
  f_{\gamma_{\rm v}}(x)&= \sum^M_{m=1}\frac{mc_m}{\alpha_{\rm bd}}e^{-\frac{m\lb(x-1\rb)}{\alpha_{\rm bd}}},
\end{align}
where $c_m=(-1)^{m+1}\binom{M}{m}$ and $\alpha_{\rm bd}\triangleq\bar{\gamma}_{\rm b}\beta_{\rm bd}$. 

\subsection{Outage Probability (OP)}
\label{subsec_OP}
The reliability of the D2D link is measured via its OP performance, which can be obtained using \eqref{eq_snr_DU_fin} and \eqref{eq_cdf_conv_square} as: 
\begin{align}\label{eq_out_d2d}
p_{\rm out}&\triangleq\mathbb{P}\lb[\Gamma_{\rm d} \le \gamma_{\rm {th}}\rb]=\int^{\infty}_0 F_{\gamma_{\rm srd}}\lb(\gamma_{\rm {th}}x\rb)f_{\gamma_{\rm v}}(x)dx \nn \\
&=\mathcal{A}
-\frac{\omega\sqrt{\beta_{\rm sd}}}{\sqrt{\beta_{\rm sd}+2\sigma^2}}
\mathcal{B},
\end{align}  
where $\gamma_{\rm {th}}\triangleq 2^{\mathcal{R}_{\rm th}}-1$ represents the SINR threshold 
with $\mathcal{R}_{\rm th}$ being the information rate requirement of the D2D link. In the latter expression, parameters $\mathcal{A}$ and $\mathcal{B}$ are given by: 
\begin{align}
\label{eq_A}
\mathcal{A}
\triangleq& 1-\omega\int^{\infty}_{0}Q\lb(\frac{\sqrt{\gamma_{\rm {th}} x}-\mu\sqrt{\bar{\gamma}_{\rm s}}}
{\sigma\sqrt{\bar{\gamma}_{\rm s}}}\rb)f_{\gamma_{\rm v}}(x)dx,\\
\label{eq_B_def}
\mathcal{B}
\triangleq& \int^{\infty}_{0}\exp\lb(-\frac{\lb(\sqrt{\gamma_{\rm {th}}x}-\mu\sqrt{\bar{\gamma}_{\rm s}}\rb)^2}
{2\beta_{\rm sd}a_t\bar{\gamma}_{\rm s}\sigma^2}\rb)  \nn \\
&\times\lb[1- Q\lb(\sqrt{\frac{2}{a_t\bar{\gamma}_{\rm s}}}
\lb(\frac{\sqrt{\gamma_{\rm {th}} x}}{\beta_{\rm sd}}
+\frac{\mu\sqrt{\bar{\gamma}_{\rm s}}}{2\sigma^2}\rb)\rb) 
\rb. \nn \\
&\lb.
-Q\lb(\sqrt{\frac{2}{a_t\bar{\gamma}_{\rm s}}}
\lb(\frac{\sqrt{\gamma_{\rm {th}} x}-\mu\sqrt{\bar{\gamma}_{\rm s}}}{2\sigma^2}\rb)\rb)\rb]f_{\gamma_{\rm v}}(x)dx. 
\end{align} 
\begin{figure*}[htp!]
\begin{small}
\begin{align}
\label{eq_A1_A2}
&\mathcal{A}=1\!-\!\omega\!\lb(\underbrace{\int^{\infty}_{\frac{\mu^2\bar{\gamma}_{\rm s}}{\gamma_{\rm {th}}}}
\sum^4_{k=1} \exp\lb[-\frac{p_k}{\sigma^2\bar{\gamma}_{\rm s}}
\lb(\sqrt{\gamma_{\rm {th}} x}-\mu\sqrt{\bar{\gamma}_{\rm s}}\rb)^2\rb]
f_{\gamma_{\rm v}}(x)dx}_{\triangleq\mathcal{A}_1} 
+\underbrace{\int^{\frac{\mu^2\bar{\gamma}_{\rm s}}{\gamma_{\rm {th}}}}_0 
\!\!\lb(1-\sum^4_{k=1} \exp\lb[-\frac{p_k}{\sigma^2\bar{\gamma}_{\rm s}}
\lb(\sqrt{\gamma_{\rm {th}} x}-\mu\sqrt{\bar{\gamma}_{\rm s}}\rb)^2\rb]\rb)
f_{\gamma_{\rm v}}(x)dx}_{\triangleq\mathcal{A}_2}\rb)
\end{align}
\end{small}
\hrule
\end{figure*}
Based on the sign of the argument of the $Q$-function in \eqref{eq_sum_Q_func}'s approximation, we split $\mathcal{A}$ into two integrals as \eqref{eq_A1_A2}. 
By substituting \eqref{eq_pdf_BS_ant_sel_bd} and using \cite[eq.(2.33.1)]{book_Gradshteyn_Ryzhik}, the integral $\mathcal{A}_1$ can be computed as in \eqref{eq_A1_fin} (top of next page), 
where $a_0=\frac{p_k}{\sigma^2}+\frac{m\bar{\gamma}_{\rm s}}{\alpha_{\rm bd}\gamma_{\rm {th}}}$, 
$b_0=-\frac{\mu p_k}{\sigma^2}$, and 
$c_0=\frac{\mu^2 p_k}{\sigma^2}-\frac{m}{\alpha_{\rm bd}}$. 
Upon rearranging the limits in \eqref{eq_A1_A2}, the integral $\mathcal{A}_2$ is rewritten as:
\begin{align}
\label{eq_A2}
\mathcal{A}_2=&\mathcal{A}_1 +\lb(1-e^{\frac{\gamma_{\rm {th}}-\mu^2\bar{\gamma}_{\rm s}}
{\alpha_{\rm bd}\gamma_{\rm {th}}}}
\rb)^M 
-\frac{\bar{\gamma}_{\rm s}}{\gamma_{\rm {th}}}\sum^{M}_{m=1}\sum^4_{k=1}  
\frac{c_m c_k}{a_0}   
 e^{-c_0}  \nn \\
&\times\lb[1-2b_0\sqrt{\frac{\pi}{a_0}}e^{\frac{b^2_0}{a_0}}
Q\lb(b_0\sqrt{\frac{2}{a_0}}\rb)\rb].
\end{align}

Similarly, $\mathcal{B}$ that appears in \eqref{eq_out_d2d} can be rewritten using \eqref{eq_B_def} as $\mathcal{B}=\mathcal{B}_1-\mathcal{B}_2$, where:
\begin{align}
\label{eq_B1_B2}
\mathcal{B}_1
=&\int^{\infty}_{0}
\exp\lb(-\frac{\lb(\sqrt{\gamma_{\rm {th}}x}-\mu\sqrt{\bar{\gamma}_{\rm s}}\rb)^2}
{2\beta_{\rm sd}a_t\bar{\gamma}_{\rm s}\sigma^2}\rb) \nn \\
&\times\lb[1-Q\lb(\sqrt{\frac{2}{a_t\bar{\gamma}_{\rm s}}}
\lb(\frac{\sqrt{\gamma_{\rm {th}} x}}{\beta_{\rm sd}}
+\frac{\mu\sqrt{\bar{\gamma}_{\rm s}}}{2\sigma^2}\rb)\rb)\rb]f_{\gamma_{\rm v}}(x)dx, \\
\mathcal{B}_2=&\int^{\infty}_{0}
\exp\lb(-\frac{\lb(\sqrt{\gamma_{\rm {th}}x}-\mu\sqrt{\bar{\gamma}_{\rm s}}\rb)^2}
{2\beta_{\rm sd}a_t\bar{\gamma}_{\rm s}\sigma^2}\rb)  \nn \\
&\times Q\lb(\sqrt{\frac{2}{a_t\bar{\gamma}_{\rm s}}}
\lb(\frac{\sqrt{\gamma_{\rm {th}} x}-\mu\sqrt{\bar{\gamma}_{\rm s}}}{2\sigma^2}\rb)\rb)
f_{\gamma_{\rm v}}(x)dx.
\end{align}
By using \eqref{eq_sum_Q_func_more}, $\mathcal{B}_1$ is derived as in \eqref{eq_B1_fin} (top of next page), where 
$a_{t1}=\frac{1}{2\sigma^2+\beta_{\rm sd}}+\frac{m\bar{\gamma}_{\rm s}}{\gamma_{\rm {th}}\alpha_{\rm bd}}$,
$b_{t1}=-\frac{\mu}{2\sigma^2+\beta_{\rm sd}}$,
$c_{t1}=\frac{\mu^2}{2\sigma^2+\beta_{\rm sd}}-\frac{m}{\alpha_{\rm bd}}$, 
$a_{t2}=a_{t1}+\frac{2p_k}{a_t\beta^2_{\rm sd}}$,
$b_{t2}=b_{t1}+\frac{4\mu p_k}{2\sigma^2+\beta_{\rm sd}}$, and
$c_{t2}=c_{t1}+\frac{\mu^2p_k}{2a_t\sigma^4}$. Following similar steps to the derivation of \eqref{eq_A1_A2}, $\mathcal{B}_2$ can be obtained as in \eqref{eq_B2_fin} (top of next page), where 
$a_{t3}=\frac{\beta_{\rm sd}p_k+\sigma^2}{\sigma^2\lb(2\sigma^2+\beta_{\rm sd}\rb)}
+\frac{m\bar{\gamma}_{\rm s}}{\gamma_{\rm {th}}\alpha_{\rm bd}}$ and
$b_{t3}=-\frac{\mu\lb(\beta_{\rm sd}p_k+\sigma^2\rb)}{\sigma^2\lb(2\sigma^2+\beta_{\rm sd}\rb)}$.

\begin{figure*}[htp!]
\begin{small}
\begin{align}
\label{eq_A1_fin}
&\mathcal{A}_1=\frac{\bar{\gamma}_{\rm s}}{\gamma_{\rm {th}}}
\sum^{M}_{m=1}\sum^4_{k=1} \frac{c_m c_k e^{-c_0}}{a_0} 
\lb[e^{-\lb(a_0\mu^2+2b_0\mu\rb)}-2b_0\sqrt{\frac{\pi}{a_0}}e^{\frac{b^2_0}{a_0}}
Q\lb(\lb(\mu a_0+b_0\rb)\sqrt{\frac{2}{a_0}}\rb)\rb] 
\end{align}
\hrule
\begin{align}
\label{eq_B1_fin}
\mathcal{B}_1&=\frac{\bar{\gamma}_{\rm s}}{\gamma_{\rm {th}}}
\sum^M_{m=1}c_m\lb[\frac{e^{-c_{t1}}}{a_{t1}}
\lb(1-2b_{t1}\sqrt{\frac{\pi}{a_{t1}}}e^{\frac{b^2_{t1}}{a_{t1}}}
Q\lb(b_{t1}\sqrt{\frac{2}{a_{t1}}}\rb)\rb)-\sum^4_{k=1}\frac{c_k e^{-c_{t2}}}{a_{t2}}
\lb(1-2b_{t2}\sqrt{\frac{\pi}{a_{t2}}}e^{\frac{b_{t2}^2}{a_{t2}}}
Q\lb(b_{t2}\sqrt{\frac{2}{a_{t2}}}\rb)\rb)\rb]
\end{align}
\hrule
\begin{align}
\label{eq_B2_fin}
\mathcal{B}_2=&\frac{\bar{\gamma}_{\rm s}}{\gamma_{\rm {th}}}\sum^M_{m=1}c_m\lb[
\frac{e^{-c_{t1}}}{a_{t1}}\lb(1-e^{-\lb(a_{t1}\mu^2+2b_{t1}\mu\rb)}
-2b_{t1}\sqrt{\frac{\pi}{a_{t1}}}e^{\frac{b^2_{t1}}{a_{t1}}}
Q\lb(b_{t1}\sqrt{\frac{2}{a_{t1}}}\rb)\rb)+\sum^4_{k=1}\frac{c_k e^{-c_{t2}}}{a_{t3}}
\lb(2 e^{-\lb(a_{t3}\mu^2+2b_{t3}\mu\rb)}-1\rb. \rb. \nn \\
&\lb. \lb.
-2b_{t3}\sqrt{\frac{\pi}{a_{t3}}}e^{\frac{b_{t3}^2}{a_{t3}}}
\lb(2Q\lb(\lb(\mu a_{t3}+b_{t3}\rb)\sqrt{\frac{2}{a_{t3}}}\rb)-Q\lb(b_{t3}\sqrt{\frac{2}{a_{t3}}}\rb)\rb)
\rb)
\rb]
\end{align}
\hrule
\end{small}
\end{figure*}
\subsection{The Case of the Absence of the Direct Link}
\label{sec_without_dir_ana}
When there is no direct link between DS and DU, the received signal expression \eqref{eq_rx_DU} deduces to:
\begin{align}
\label{eq_rx_DU_nodir}
y_{\rm s}&=\sqrt{P_{\rm s}}\mathbf{h}^H_{\rm rd}\mathbf{\Theta} \mathbf{h}_{\rm sr}s_{\rm s} + 
\mathbf{w}^H_{\rm b}\mathbf{h}_{\rm bd}s_{\rm b} +n_{\rm s}, 
\end{align}
Additionally, from \eqref{eq_snr_DU_fin}, $\gamma_{\rm srd}=\bar{\gamma}_{\rm s}\abs{X}^2$ with 
$\theta^{\star}_n=-\lb(\angle h_{\text{sr},n}+\angle h_{\text{rd},n}\rb)$ $\forall n$. Therefore, the CDF of $\gamma_{\rm srd}$ in \eqref{eq_cdf_conv_square} reduces to the following simple expression:
\begin{align}
 F_{\gamma_{\rm srd}}(x)=1-\omega Q\lb(\frac{\sqrt{x/\bar{\gamma}_{\rm s}}-\mu}{\sigma}\rb),
\end{align}
which can be substituted in \eqref{eq_out_d2d} to obtain OP via \eqref{eq_A1_A2}, i.e.:  
\begin{small}
\begin{align}
\label{eq_out_d2d_nodir}
p_{\rm out}&=\int^{\infty}_0 F_{\gamma_{\rm srd}}\lb(\gamma_{\rm {th}}x\rb)f_{\gamma_{\rm v}}(x)dx 
=\mathcal{A}.
\end{align} 
\end{small}
\begin{figure*}[btp!]
 \begin{small}
\begin{align}
\label{eq_d2d_snr_2ndorder1}
\widehat{\Gamma}_{\rm d}&=
\frac{2P_{\rm s}d_0^{\eta}}{P_{\rm b}d_{\rm sd}^{\eta}}
\lb(1+\frac{N\alpha\pi\sqrt{\pi d_{\rm sd}^{\eta}\beta_{\rm sr}\beta_{\rm rd}}}{4\sqrt{d_0^{\eta}}}
+\frac{N^2\alpha^2 d_{\rm sd}^{\eta}\beta_{\rm sr}\beta_{\rm rd}}{d_0^{\eta}}\rb)
\frac{\sum^M_{m=1}\frac{c_m e^{\frac{mN_0}{P_{\rm b}\beta_{\text{bd}}}}}
{\lb(m/\beta_{\text{bd}}\rb)^2}}
{\lb(\sum^M_{m=1}\frac{c_m e^{\frac{mN_0}{P_{\rm b}\beta_{\text{bd}}}}}{m/\beta_{\text{bd}}}\rb)^3}
\end{align}
 \end{small}
\hrule
\end{figure*}
\section{Proposed Optimization Framework}
\label{sec_joint_opt}
In this section, we present our optimization framework for the design of the proposed RIS-empowered underlaid D2D system. To deal with the derived complicated OP performance expression, we present efficient approximations and an equivalent transformation of the design problem formulation.

\subsection{Optimization Formulation}
We now formulate the proposed OP minimization problem including the following constraints: 
the interference control for CU, the individual power budgets at DS and BS, and the RIS deployment (RD) constraint under the parallel topology placement; in the next section, the elliptical topology will be also considered. 
It is noted that, to adhere to the underlay definition of the considered cognitive communications, 
we next use an interference temperature constraint. However, the presented framework is also valid for a CU rate 
constraint. 
The mathematical definition of our OP minimization problem is : 
\begin{equation}
\label{eq_outg_min_P0}
\begin{aligned}
(\mathbf{P}_{\mathrm{0}}):\min_{P_{\rm s}, d} ~~ & p_{\rm out}\lb(P_{\rm s}, d\rb) \quad\text{subject to (s.t.)} \nn \\
 \textstyle({\rm C1}):& I_{\rm b}\lb(P_{\rm s}\rb) - I_{\textrm{th}}\le 0, \nn \\
 \textstyle({\rm C2}):& P_{\rm s}\ge 0,~~~\qquad\textstyle({\rm C3}): P_{\rm s} \le P^{\textrm{max}}_{\rm s},\nn \\  
 \textstyle({\rm C4}):& d\ge \sqrt{\delta^2-y^2}, 
 \textstyle({\rm C5}): d\le d_{\rm sd}- \sqrt{\delta^2-y^2},
\end{aligned}
\end{equation}
where $\delta\triangleq 2fL^2/c$ indicates the smallest distance that needs to exist between the DS and RIS 
(or RIS and DU) \cite{DM_ICC_2016} with
$f$ denoting the frequency of the transmitted signal, while $L$ and $c$ 
represent the radius of antenna and the speed of light, respectively. It can be easily recognized that it is difficult to directly solve $\mathbf{P}_{\mathrm{0}}$, since it involves the coupled optimization 
variables $P_{\rm s}$ and $d$, and a non-convex objective function. Also, the $p_{\rm out}$ expression includes the 
$Q$-function and multiple summation terms, which further complicates the solution. 
In the following, we present a simplified objective function that results in an optimization problem adhering to a globally optimal solution, which will help us draw useful insights. 

\subsection{Proposed Approximations}
We present two approximations that will be used to derive simpler expressions for $p_{\rm out}$ in $\mathbf{P}_{\mathrm{0}}$ and the moments of 
\subsubsection{Expectation of $\Gamma_{\rm d}$}
\label{subsubsec_approx_Taylor}
Since the OP expression in \eqref{eq_out_d2d} is of an integral form, we perform a transformation analysis on $\Gamma_{\rm d}$. In particular, we deploy Jensen's inequality to derive the following approximation for $\Gamma_{\rm d}$. 
\begin{lemma}
\label{approx_gamma_d}
The function $\Gamma_{\mathrm{d}}(\cdot)$ in \eqref{eq_snr_DU_fin} is unimodal, hence, using Jensen's inequality, we can write 
$\mathbb{E}[\Gamma_{\mathrm{d}}(\cdot)]\ge \Gamma_{\mathrm{d}}(\mathbb{E}[\cdot])$. 
An approximation of the expected value of $\Gamma_{\rm d}$ is given in expression \eqref{eq_d2d_snr_2ndorder1} (top of this page). 
\end{lemma}
\begin{proof}
Using the expansion of the second-order Taylor series, the function $\Gamma_{\mathrm{d}}(\cdot)$ can be approximated around the point 
$(\mathbb{E}\lb[\gamma_{\rm srd}\rb],\mathbb{E}\lb[\gamma_{\rm v}\rb])$ as \cite{PDF_taylor_approx}:
\begin{align}
\label{eq_d2d_snr_2ndorder}
\widehat{\Gamma}_{\rm d}&=\mathbb{E}\lb[\frac{\gamma_{\rm srd}}{\gamma_{\rm v}}\rb] 
\nn \\&
\approx\frac{\mathbb{E}\lb[\gamma_{\rm srd}\rb]}{\mathbb{E}\lb[\gamma_{\rm v}\rb]}
-\frac{\cov(\gamma_{\rm srd},\gamma_{\rm v})}{(\mathbb{E}\lb[\gamma_{\rm v}\rb])^2}
+\frac{\mathbb{V}\lb[\gamma_{\rm v}\rb]\mathbb{E}\lb[\gamma_{\rm srd}\rb]}{(\mathbb{E}\lb[\gamma_{\rm v}\rb])^3}~.
\end{align}
Since $\gamma_{\rm srd}$ and $\gamma_{\rm v}$ are independent random variables, it holds $\cov(\gamma_{\rm srd},\gamma_{\rm v})=0$. 
Their expected values in \eqref{eq_snr_DU_fin} and the second moment of $\gamma_{\rm v}$ can be respectively computed as follows:
\begin{align}
\label{eq_mean_srd}
\mathbb{E}\lb[\gamma_{\rm srd}\rb]&
=\frac{P_{\rm s}\beta_{\rm sd}}{N_0}\lb(1+\frac{N\alpha\pi\sqrt{\pi\beta_{\rm sr}\beta_{\rm rd}}}{4\sqrt{\beta_{\rm sd}}}
+N^2\alpha^2\frac{\beta_{\rm sr}\beta_{\rm rd}}{\beta_{\rm sd}}\rb), \\
\label{eq_mean_v}
\mathbb{E}\lb[\gamma_{\rm v}\rb]&=
\frac{P_{\rm b}}{N_0}\sum^M_{m=1}\frac{c_m}{m/\beta_{\text{bd}}} 
e^{\frac{mN_0}{P_{\rm b}\beta_{\text{bd}}}}
, \\
\label{eq_second_moment_v}
\mathbb{E}\lb[\gamma^2_{\rm v}\rb]&=
2\lb(\frac{P_{\rm b}}{N_0}\rb)^2 \sum^M_{m=1}\frac{c_m
e^{\frac{mN_0}{P_{\rm b}\beta_{\text{bd}}}}}
{\lb(m/\beta_{\text{bd}}\rb)^2}. 
\end{align}
Furthermore, using \eqref{eq_mean_v} and \eqref{eq_second_moment_v}, $\mathbb{V}\lb[\gamma_{\rm v}\rb]$ can be derived as: 
\begin{align}
\label{eq_var_gamma_v}
&\mathbb{V}\lb[\gamma_{\rm v}\rb]=\mathbb{E}\lb[\gamma^2_{\rm v}\rb]-\lb(\mathbb{E}\lb[\gamma_{\rm v}\rb]\rb)^2 
\nn \\  &
=\lb(\frac{P_{\rm b}}{N_0}\rb)^2\lb(2\sum^M_{m=1}\frac{c_m e^{\frac{mN_0}{P_{\rm b}\beta_{\text{bd}}}}}
{\lb(m/\beta_{\text{bd}}\rb)^2}-\lb(\sum^M_{m=1}\frac{c_m e^{\frac{mN_0}{P_{\rm b}\beta_{\text{bd}}}}}
{m/\beta_{\text{bd}}}\rb)^2\rb).
\end{align}
simplified expression of $\widehat{\Gamma}_{\rm d}$ in \eqref{eq_d2d_snr_2ndorder1} is derived. 
\end{proof}
\subsubsection{Moments of $\gamma_{\rm bd}$}
\label{subsubsec_gumbel_approx_gamma_bd}
In order to obtain simplified versions of \eqref{eq_mean_v} and \eqref{eq_var_gamma_v}, we make use of the Gumbel distribution to approximate the moments of $\gamma_{\rm bd}$, and in turn $\gamma_{\rm v}$. Let 
$X^{\star}_m$ be the maximum among the $M$ samples $\{\gamma_{\text{bd},1},\gamma_{\text{bd},2},\ldots,\gamma_{\text{bd},M}\}$. Then, for any arbitrary sequence of $a_m$'s and $b_m$'s, the following exact probability is calculated:
\begin{align}
\mathbb{P}\lb[\frac{X^{\star}_m-a_m}{b_m} \le z\rb] =&\lb(\mathbb{P}\lb[X_1\le b_mz+a_m\rb]\rb)^M \nn\\
=&\lb[1-\exp\lb(-\frac{1}{\alpha_{\rm bd}}\lb(b_m z+a_m\rb)\rb)\rb]^M. 
\end{align}
By setting $a_m=\alpha_{\rm bd}\log M$ and $b_m=\alpha_{\rm bd}$ as $M\rightarrow\infty$, yields $\mathbb{P}\lb[\frac{X^{\star}_m-a_m}{b_m} \le z\rb]=\exp\lb(-e^{-z}\rb)$. 
Hence, the first and the second moment as well as the variance of $X^{\star}_m$ are respectively given as follows: 
\begin{align}
\label{eq_max_var_moments}
\mathbb{E}\lb[X^{\star}_m\rb]&=\alpha_{\rm bd}\lb(\log M+\zeta \rb), \\
\mathbb{E}\lb[(X^{\star}_m)^2\rb]&=\alpha^2_{\rm bd}\lb(\log M+\zeta \rb)^2+\alpha^2_{\rm bd}\pi^2/6,\nn \\
\mathbb{V}\lb[X^{\star}_m\rb]&=\mathbb{E}\lb[(X^{\star}_m)^2\rb]-\lb(\mathbb{E}\lb[X^{\star}_m\rb]\rb)^2
=\alpha^2_{\rm bd}\pi^2/6,
\end{align}
where $\zeta$ is the Euler- Mascheroni constant with $\zeta\approx0.5772$ \cite{Euler_number}.

Finally, it holds $\gamma_{\rm v}=1+\gamma_{\rm bd}$, thus, the mean and the variance of $\gamma_{\rm v}$ are given by:
\begin{align}
\label{eq_transformed_1stmoment}
\mathbb{E}\lb[\gamma_{\rm v}\rb]&=1+\mathbb{E}\lb[X^{\star}_m\rb]=1+\alpha_{\rm bd}\lb(\log M+\zeta \rb), \\
\label{eq_transformed_var}
\mathbb{V}\lb[\gamma_{\rm v}\rb]&=\mathbb{V}\lb[X^{\star}_m\rb]=\alpha^2_{\rm bd}\pi^2/6.
\end{align}

\subsection{Equivalent Transformation of the Problem $\mathbf{P}_{\mathrm{0}}$}
We next present a suboptimal solution for the power allocation $P_{\rm s}^{\star}$ at DS and the RIS
deployment $d^{\star}$
that maximizes $\widehat{\Gamma}_{\rm d}$, subject to the interference level constraint $\textstyle({\rm C1})$, the DS transmit power 
constraints $\textstyle({\rm C2})$ and $\textstyle({\rm C3})$, and the RD constraints $\textstyle({\rm C4})$ and $\textstyle({\rm C5})$:
\begin{equation}
\label{eq_sinr_max_P2}
\begin{aligned}
(\mathbf{P}_{\mathrm{1}}):\max_{P_{\rm s},d} \quad & 
{\widehat{\Gamma}_{\rm d}}\lb(P_{\rm s},d\rb)\,\, \textrm{s.t.~}\,\, \rm \textstyle({\rm C1})-\textstyle({\rm C5}). 
 \end{aligned}
\end{equation}
We first prove that the optimal solution for $\mathbf{P}_{\mathrm{0}}$ is feasible for $\mathbf{P}_{\mathrm{1}}$ using the following lemma. 
\begin{lemma}
\label{lemma_OP_min}
Minimizing the OP of the D2D communication pair is equivalent to maximizing the expected value of $\Gamma_{\rm d}$. 
\end{lemma}
\begin{proof}
\label{proof_lemma_OP_min}
Expression \eqref{eq_out_d2d} can be rewritten as follows: 
\begin{align}
p_{\rm out}=1-\overline{F_{\Gamma_{\rm d}}}(\gamma_{\rm {th}})=1-\mathbb{P}\lb[\Gamma_{\rm d} > \gamma_{\rm {th}}\rb]. 
\end{align}
In addition, it will be shown that $\overline{F_{\Gamma_{\rm d}}}(\cdot)$ is an increasing function of its mean 
$\mathbb{E}\lb[\Gamma_{\rm d}\rb]$. 
Let us denote the jointly optimal solution maximizing $\Gamma_{\rm d}$ in   \eqref{eq_snr_DU_fin} as $\Gamma^{\star}_{\rm d}$. 
Then, it holds that: 
\begin{align}
\label{eq_joint_opt_compare}
\Gamma^{\star}_{\rm d}& > \Gamma_{\rm d} \nn \\
\Rightarrow\mathbb{E}\lb[\Gamma^{\star}_{\rm d}\rb]& > \mathbb{E}\lb[\Gamma_{\rm d}\rb],
\end{align}
with $\Gamma_{\rm d}\neq \Gamma^{\star}_{\rm d}$.  
By using \cite{Garg_outg_prob_min_transformpaper}, yields the inequality: 
\begin{align}
\mathbb{P}\lb[\log_2\lb(1+\mathbb{E}\lb[\Gamma^{\star}_{\rm d}\rb]\rb)\le \gamma_{\rm {th}}\rb]
& < \mathbb{P}\lb[\log_2\lb(1+\mathbb{E}\lb[\Gamma_{\rm d}\rb]\rb)\le \gamma_{\rm {th}}\rb],
\end{align}
which essentially means that $\Gamma^{\star}_{\rm d}$ also solves $\mathbf{P}_{\mathrm{0}}$. In conclusion, minimizing the CDF is equivalent to maximizing the expected value of the same random variable. This holds true for all 
unimodal distributions. 
\end{proof}

\section{Jointly Optimal RIS Placement and DS Power Allocation}
\label{sec_opt_sol}
The optimization problem $\mathbf{P}_{\mathrm{1}}$ includes the two considered design variables: the DS transmit power $P_{\rm s}$ and the RIS location parameter $d$. It will be next shown that the optimization over these variables can be carried out individually.
\begin{lemma}
 \label{lemma_decouple_vars}
The design problem $\mathbf{P}_{\mathrm{1}}$ can be decoupled into two sub-problems: one optimizing over $d$ and the other over $P_{\rm s}$. 
\end{lemma}
\begin{proof}
The objective function $\widehat{\Gamma}_{\rm d}$ of $\mathbf{P}_{\mathrm{1}}$ is strictly increasing with $P_{\rm s}$ for any $d$ value. After obtaining the optimal location, the optimal power allocation 
can be found by computing the maximum value of $P_{\rm s}$ that satisfies the constraint $\textstyle({\rm C1})$. We, thus, conclude that $\mathbf{P}_{\mathrm{1}}$ can be decoupled into two sub-problems. The first sub-problem involves optimizing the location for any feasible $P_{\rm s}$ value, which is followed by the second sub-problem that seeks finding the optimal $P_{\rm s}$ value for any feasible value of $d$. 
\end{proof}

\subsection{Optimal RIS Placement}\label{sec_opt_RIS}
As discussed in Lemma \ref{lemma_decouple_vars}, to derive the optimal solution for $\mathbf{P}_{\mathrm{1}}$, 
we first obtain the optimal RD, which is then followed by the optimal DS power allocation. The following lemma includes a key result that will help 
in simplifying the process of obtaining the optimal RD. 

\begin{lemma}
\label{lemma_parallel}
Maximizing $\widehat{\Gamma}_{\rm d}$ over the variable $d$, considering the RD constraints $\textstyle({\rm C4})$ and $\textstyle({\rm C5})$, 
is equivalent to minimizing the function $Z_p\triangleq\lb(y^2+d^2\rb)\lb(y^2+\lb(d_{\rm sd}-d\rb)^2\rb)$.
\end{lemma}
\begin{proof}
Using the definition of $Z_p$ and by substituting \eqref{eq_beta_parallel} in \eqref{eq_d2d_snr_2ndorder1}, 
$\widehat{\Gamma}_{\rm d}$ can be rewritten as follows:
\begin{align}
 \widehat{\Gamma}_{\rm d}=&
\frac{2P_{\rm s}d_0^\eta}{P_{\rm b}d_{\rm sd}^{\eta}}
\lb[1+\frac{N\alpha\pi d_0^{\eta/2}d_{\rm sd}^{\eta/2}\sqrt{\pi}}
{4Z_p^{\eta/4}} 
+\frac{N^2\alpha^2 d_0^{\eta}d_{\rm sd}^\eta}{Z_p^{\eta/2}}\rb]  \nn \\
 &\times\frac{\sum^M_{m=1}\frac{c_m e^{\frac{mN_0}{P_{\rm b}\beta_{\text{bd}}}}}
{\lb(m/\beta_{\text{bd}}\rb)^2}}
{\lb(\sum^M_{m=1}\frac{c_m e^{\frac{mN_0}{P_{\rm b}\beta_{\text{bd}}}}}{m/\beta_{\text{bd}}}\rb)^3}. 
\end{align}
We can, thus, express $\widehat{\Gamma}_{\rm d}$ as a function of $d$: 
\begin{align}
\widehat{\Gamma}_{\rm d}(Z_p, d)&=\widehat{\Gamma}_{\rm d}(Z_p(d))=\widehat{\Gamma}_{\rm d}\circ Z_p.
\end{align}
This transformation is monotonic, hence, according to \cite[Preposition~1]{DM_ICC_2016}, maximizing $\widehat{\Gamma}_{\rm d}$ over $d$ is 
equivalent to minimizing $Z_p$ over the same variable. 
\end{proof}
Following Lemma \ref{lemma_parallel}, an equivalent optimization problem for RD is obtained:
\begin{equation}
\label{eq_sinr_max_RD}
\begin{aligned}
(\mathbf{P}_{\mathrm{RD}}):\min_{d} \quad & 
{Z_p}\lb(d\rb) \,\,\textrm{s.t.~} \,\,\rm \textstyle({\rm C4}),\,\textstyle({\rm C5}),
 \end{aligned}
\end{equation}
with $\textstyle({\rm C4})$ and $\textstyle({\rm C5})$ being boundary constraints.  
Since the objective function is not convex in $d$, to compute the optimal solution, we first find the critical point for $Z_p$ with respect to $d$, by solving the following equation: 
\begin{align}
\label{eq_diff_zero_parallel}
\frac{\partial Z_p(d)}{\partial d}&=2\lb(2d-d_{\rm sd}\rb)\lb(d^2-d d_{\rm sd}+y^2\rb)=0. 
\end{align}
This equation has three solutions for $d$, namely, 
$d_1=d_{\rm sd}/2$ and $\{d_2,d_3\}=\lb(d_{\rm sd}\pm\sqrt{d_{\rm sd}^2-4y^2}\rb)/2$. Therefore, the set of candidate points for the optimal $d$ comprises the solutions of \eqref{eq_diff_zero_parallel} along with the boundary conditions $\textstyle({\rm C4})$ (namely, $d_4=\sqrt{\delta^2-y^2}$) 
and $\textstyle({\rm C5})$ (namely, $d_5=d_{\rm sd}-\sqrt{\delta^2-y^2}$). The second derivative of $Z_p$ over $d$ can be derived as follows:
\begin{align}
\label{eq_diff_order2_parallel}
\frac{\partial^2 Z_p(d)}{\partial d^2}&=2\lb(6d^2-6d d_{\rm sd}+2y^2+d^2_{\rm sd}\rb),
\end{align} 
hence, its values at the candidate $d$-points are given by:
\begin{equation}
\mathcal{I}=
\begin{cases}
 \label{second_derivative_parallel}
-\lb(d^2_{\rm sd}-4y^2\rb)<0, ~d=d_1 \\
2\lb(d^2_{\rm sd}-4y^2\rb)>0, ~d=\{d_2,d_3\} \\
2\lb(6\delta^2-6d_{\rm sd}\sqrt{\delta^2-y^2}+d^2_{\rm sd}-4y^2\rb)>0, \\
d=\{d_4,d_5\}
\end{cases}\,\,.
\end{equation}
As can be seen from the latter expression, $d_1$ needs to be excluded from the candidate set of \eqref{eq_diff_zero_parallel}'s solutions, since $I_1$ is negative. However, the remaining $d$-points can be substituted in the objective function of problem $\mathbf{P}_{\mathrm{RD}}$ to find the one solving it. The  values of $Z_p(\cdot)$ for these four candidate points can be calculated as follows: 
\begin{equation}
 \tilde{Z}_p(d)=
 \begin{cases}
 \label{Z_cond1}
\frac{1}{4}\lb(d^2_{\rm sd}\pm d_{\rm sd}\sqrt{d^2_{\rm sd}-4y^2}\rb)^2, ~d=\{d_2,d_3\} \\
\delta^2\lb(d^2_{\rm sd}+\delta^2-2d_{\rm sd}\sqrt{\delta^2-y^2}\rb), ~d=\{d_4,d_5\}
\end{cases}.
\end{equation}
Since it holds that $\delta^2\lb(d^2_{\rm sd}+\delta^2-2d_{\rm sd}\sqrt{\delta^2-y^2}\rb)
<\lb(d^2_{\rm sd}\pm d_{\rm sd}\sqrt{d^2_{\rm sd}-4y^2}\rb)^2$, the optimal value of $d$ is: 
\begin{align}
\label{eq_d_opt_parallel}
 d^{\ast}\triangleq \{d_4,d_5\}. 
\end{align}
 
We next solve the RD problem for the elliptical topology.  
\begin{corollary}
\label{ellipse_corollary}
The maximization of $\widehat{\Gamma}_{\rm d}$ in terms of $d$ is equivalent to minimizing $Ze\triangleq d^2\lb(d_{\rm sd}/\epsilon-d\rb)^2$ over the same variable.
\end{corollary}
\begin{proof}
Substituting the definitions of $\beta_{\rm sr}$ and $\beta_{\rm rd}$ in \eqref{eq_beta_ellipse} into \eqref{eq_d2d_snr_2ndorder1} 
and following similar steps to those included in Lemma~\ref{lemma_parallel}, yields the proof. 
\end{proof}
By using Corollary~\ref{ellipse_corollary}, we formulate the following optimization problem for RD: 
\begin{equation}
\label{eq_sinr_min_P3}
\begin{aligned}
(\textbf{P2}):\min_{d} \quad & 
{Z_e}\lb(d\rb) \\
\textrm{s.t.}\,\,\textstyle({\rm C6}):&\,\, d\ge \delta, ~~\textstyle({\rm C7}):& d \le d_{\rm sd}/\epsilon-\delta. 
 \end{aligned}
\end{equation}
where $\textstyle({\rm C6})$ and $\textstyle({\rm C7})$ are boundary constraints. 
To obtain the critical point of $Z_e(\cdot)$ with respect to $d$, the following equation needs to be solved 
\begin{align}
\label{eq_diff_zero_ellipse}
\frac{\partial Z_e(d)}{\partial d}&=2d\lb(d_{\rm sd}/\epsilon-d\rb)\lb(d_{\rm sd}/\epsilon-2d\rb)=0.
\end{align} 
The boundary conditions $\textstyle({\rm C6})$ and $\textstyle({\rm C7})$, along 
with the solutions of \eqref{eq_diff_zero_ellipse}, form the set of candidate points for the optimal $d$. 
From these conditions, we conclude that $d=0, d_{\rm sd}/\epsilon$ are not valid 
solutions. Next, we check the second derivative of $Z_e(\cdot)$ over $d$ to find the maximum point:
\begin{align}
\label{eq_diff_ellipse_2ndorder}
\frac{\partial^2 Z_e(d)}{\partial d^2}&=2\lb[\lb(d_{\rm sd}/\epsilon-2d\rb)^2-2d\lb(d_{\rm sd}/\epsilon-d\rb)\rb].
\end{align}
Evidently, the three candidate RD points are 
$d=d_{\rm sd}/(2\epsilon)$, $\delta$, $d_{\rm sd}/\epsilon-\delta$ leading to the following values for $\textbf{P2}$'s objective: 
\begin{equation}
 \tilde{Z}_e(d)=
 \begin{cases}
 \label{Z_cond1_ellip}
\frac{d^4_{\rm sd}}{16\epsilon^4}, &d=d_{\rm sd}/(2\epsilon) \\
\delta^2\lb(d_{\rm sd}/\epsilon-\delta\rb)^2, &d=\delta, d_{\rm sd}/\epsilon-\delta
\end{cases}.
\end{equation}
Therefore, the optimal value of $d$ is $d^{\star}=\{\delta, d_{\rm sd}/\epsilon-\delta\}$.

\subsection{Optimal DS Power Allocation}\label{sec_opt_power}
As shown in \eqref{eq_d2d_snr_2ndorder1}, the function $\widehat{\Gamma}_{\rm d}(\cdot)$ increases monotonically with respect to $P_{\rm s}$. 
However, the constraint $\textstyle({\rm C1})$ restricts $P_{\rm s}$ to be above a certain value in 
order to ensure that the D2D-induced interference to CU is lower than a predefined threshold  $I_{\textrm{th}}$.
Hence, the optimal value of $P_{\rm s}$ needs to constrained as follows: 
\begin{align}
\label{eq_ub_Ps}
P^{\rm ub}_{\rm s}&=I^{\textrm{th}}/\mathbb{E}\lb[\lvert h_{\rm sc}\rvert^2\rb].
\end{align}
By using this expression, we conclude that the optimal solution for DS power allocation $P_{\rm s}$ is:
\begin{align}
\label{eq_Ps_opt}
P^{\star}_{\rm s}&\triangleq\min\{P^{\textrm{max}}_{\rm s}, P^{\rm ub}_{\rm s}\}. 
\end{align}

\section{Numerical Results and Discussion} 
\label{sec_results}
In this section, we present computer simulation results that validate the presented analytical expressions and provide useful insights on the proposed optimized RIS-empowered underlaid D2D communication system.

\subsection{Simulation Parameters}\label{sec_sim_param}
Unless otherwise mentioned, we have considered the following system parameters in the performance evaluations: 
$N=50$, $\alpha=0.5$, $M=1$, $d_0=1.05$ m, $d_{\rm sr}=d_{\rm rd}=1.5$m \cite{Ni_Zhu_IRS_D2D_nakagami_perfana}, 
$d_{\rm bd}=300$m, $d_{\rm sc}=250$m, $\eta=2.5$, $d_{\rm sd}=5$ m, $\delta=0.75$ m, $y=0.5$m, $\epsilon=1$ \cite{DM_ICC_2016}, 
$\bar{\gamma}_{\rm b}=28$ dB, $I_{\rm th}=11$ dB, and $P^{\textrm{max}}_{\rm s}=10$ dB, $C$= -$30$ dB at reference distance $1$ m.  
We have considered an omni-directional antenna at the BS. For simulation purposes, the gain of the
transmit and receive antenna are assumed to be unity.

For comparison purposes, we have also simulated the performance of the fixed allocation scheme where $P_{\rm s}=P^{\textrm{max}}_{\rm s}-5$ and $d=d_{\rm sd}-1.5$ m. All numerical results have been obtained 
by averaging over $10^6$ independent fading channel realizations.
\subsection{Validation of the Analysis} \label{sec_result_valid} 
The numerical evaluation of the analytical approximation in \eqref{eq_d2d_snr_2ndorder1} for $\widehat{\Gamma}_{\rm d}$ is illustrated in Fig.~\ref{FIG_valid_plot_approx_SINR_N} varying $\bar{\gamma}_{\rm s}$ in dB for different values of $N$ of the RIS elements. The figure also includes simulation results for the exact $\Gamma_{\rm d}$, which were generated by computing the mean of $10^6$ random realizations of this statistics using \eqref{eq_cdf_conv_square}.  
It can be observed in the figure that the considered Taylor series approximation for the SINR at DU closely matches the equivalent simulated results. In fact, the approximation becomes tight for larger values of $N$. Thus, the analysis provided in Section \ref{subsubsec_approx_Taylor} is validated. Recall from Lemma \ref{lemma_OP_min} that minimizing the OP of the D2D communication pair is equivalent to maximizing the expected value of $\Gamma_{\rm d}$. It is also observed in the figure that, when $N$ increases, the SINR at DU improves. This happens because more RIS reflecting elements result in more reflected signals, which undergo constructive interference at DU. 
     \begin{figure}[!t]
     \centering
         \includegraphics[trim=3cm 9cm 3.9cm 9.5cm, width=0.9\columnwidth, clip]{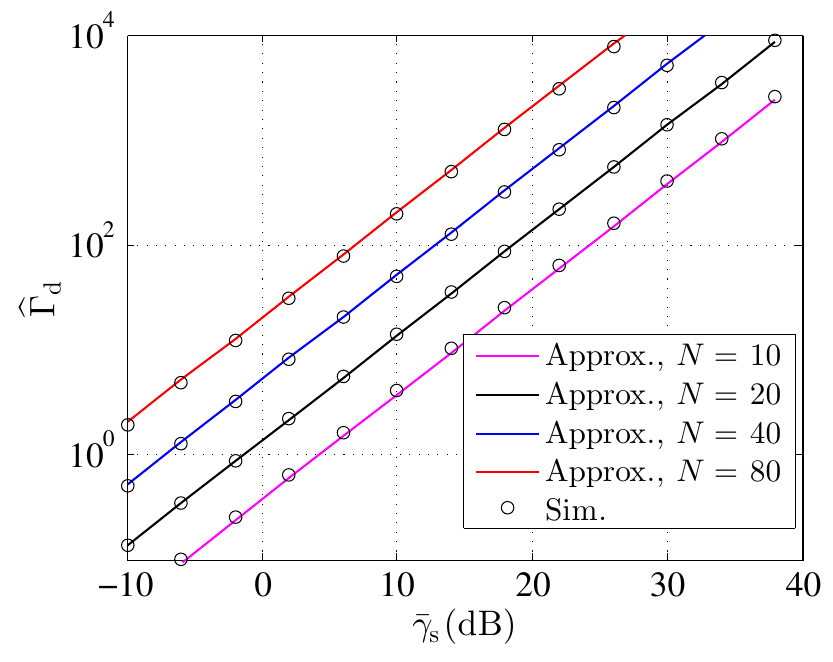}
 \caption{The values of $\widehat{\Gamma}_{\rm d}$ as a function of $\bar{\gamma}_{\rm s}$ in dB considering an RIS with $N$=\{10, 20, 40, 80\} tunable reflecting elements.}         \label{FIG_valid_plot_approx_SINR_N}
\end{figure}
\begin{figure}[!t]
     \centering
         \includegraphics[trim=3cm 9cm 3.8cm 9.5cm, width=0.9\columnwidth, clip]{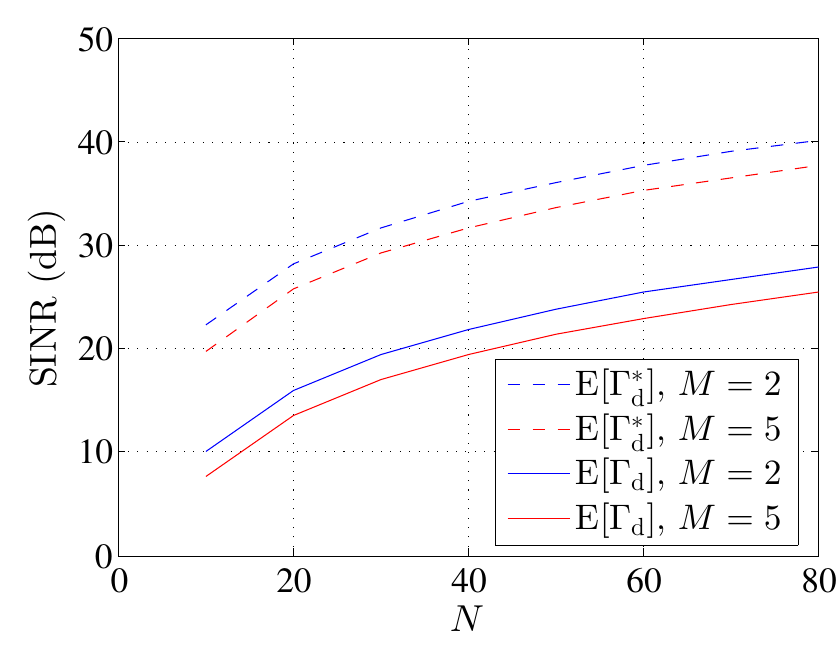}
    \caption{The values of $\Gamma_{\mathrm{d}}$ and $\Gamma_{\mathrm{d}}^{\star}$ as a function of $N$.} \label{FIG_valid_plot_comp_SINR_N}
   \end{figure} 
   
A numerical comparison of the average values of $\Gamma_{\mathrm{d}}$ and $\Gamma_{\mathrm{d}}^{\star}$ for 
different values of $N$, considering $M=\{2,5\}$, is provided in Fig.~\ref{FIG_valid_plot_comp_SINR_N}. The figure includes average value of $10^6$ random realizations of the objective function $\widehat{\Gamma}_{\mathrm{d}}$ using \eqref{eq_d2d_snr_2ndorder1} and the joint optimal solution $\Gamma_{\mathrm{d}}^{\star}$ corresponding to \eqref{eq_Ps_opt}, \eqref{eq_d_opt_parallel}. 
It is observed that $\Gamma_{\mathrm{d}}^{\star}$ outperforms $\Gamma_{\mathrm{d}}$. Also, an increase in the number of antennas at the BS, degrades the SINR performance. Thus the claim in \eqref{eq_joint_opt_compare} is validated. 

\begin{figure}[!t] 
	\centering\hspace{-3mm}
	\subfigure[Mean of $\gamma_{\rm v}$ vs. $M$.]{
		\includegraphics[trim=3cm 9cm 3.8cm 9.5cm, width=0.5\columnwidth, clip]{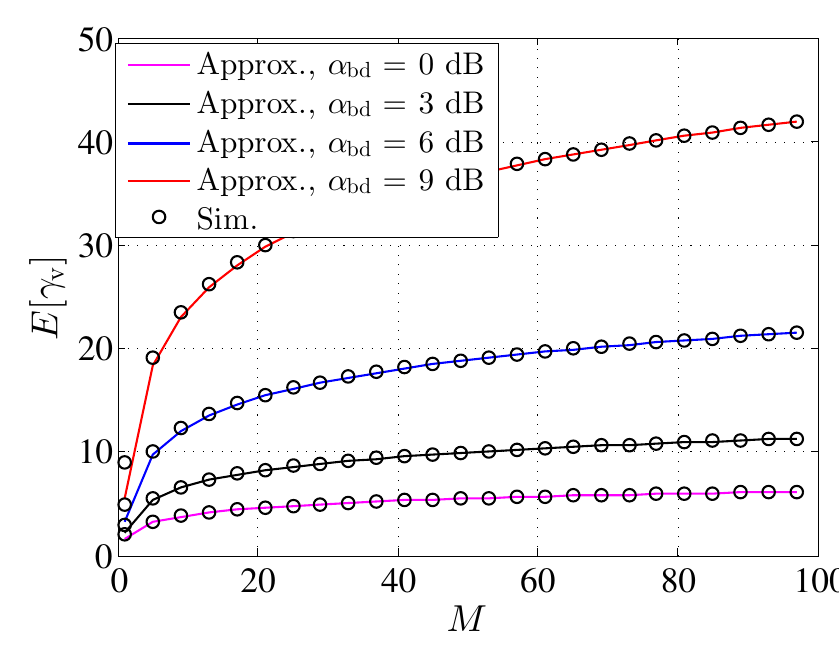} 
		\label{Fig3_validation_plot_gumb_mean}
	}\hspace{-6mm}
	\subfigure[Variance of $\gamma_{\rm v}$ vs. $M$.]{
		\includegraphics[trim=3cm 9cm 3.8cm 9.5cm, width=0.5\columnwidth, clip]{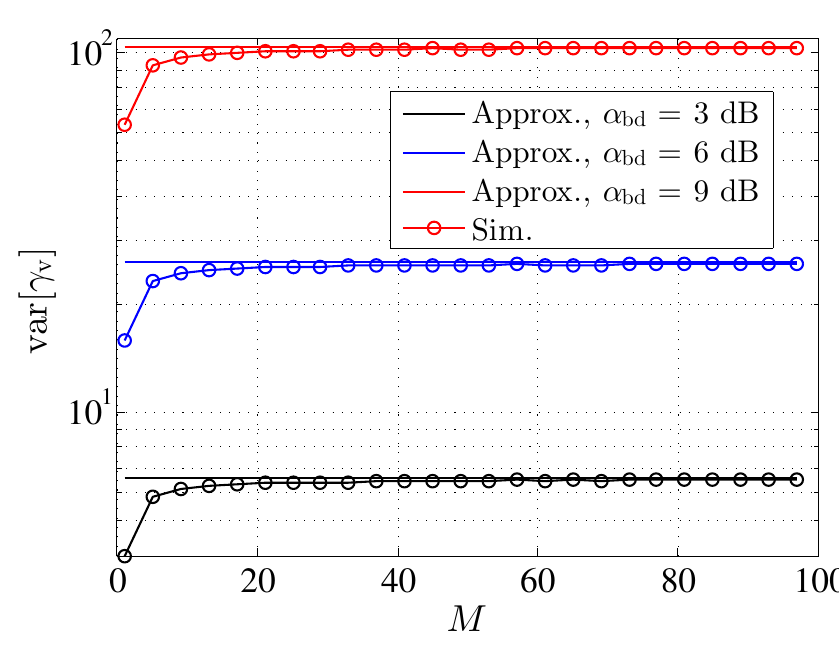}\hspace{-5mm}
		\label{Fig4_validation_plot_gumb_var}
	}
	\caption{The mean and variance of the Gumbel-distributed random variable $\gamma_{\rm v}$ as functions of the number $M$ of BS antennas for different values of $\alpha_{\rm bd}$ in dB.}
\label{Fig_validation_plot_gumb}	
\end{figure}       
In Fig. \ref{Fig_validation_plot_gumb}, we validate the asymptotic mean and variance of the Gumbel-distributed random variable $\gamma_{\rm v}$, as presented in 
Section~\ref{subsubsec_gumbel_approx_gamma_bd}, as functions of the number of BS antennas $M$. In Fig. \ref{Fig3_validation_plot_gumb_mean}, we evaluate 
$\mathbb{E}[\gamma_{\rm v}]$ using \eqref{eq_transformed_1stmoment} for four different values 
of the $\alpha_{\rm bd}$ parameter. It is shown that, as $M$ approaches infinity, the simulated results for the mean closely match with the mean of the Gumbel distribution. Moreover, it is observed that the expected value of $\gamma_{\rm v}$ increases with increasing values of $\alpha_{\rm bd}$. 
In Fig. \ref{Fig4_validation_plot_gumb_var}, we plot ${\rm Var}[\gamma_{\rm v}]$, via the numerical evaluation of \eqref{eq_transformed_var}, versus $M$. 
It is demonstrated that the simulated variance initially improves with increasing $M$, and then, slowly saturates as $M$ approaches infinity. The simulated results match 
closely with the provided approximation for different values of $\alpha_{\rm bd}$. We 
also observe that, for smaller values of $\alpha_{\rm bd}$, the saturation occurs at smaller values of $M$. 
However, as $\alpha_{\rm bd}$ increases, the saturation takes place at larger $M$ values. This behavior is attributed to the resulting improved beamforming gain in the BS-DU link.   

\begin{figure}[!t]
\centering
\includegraphics[trim=3cm 9cm 3.7cm 9.5cm, width=0.9\columnwidth, clip]{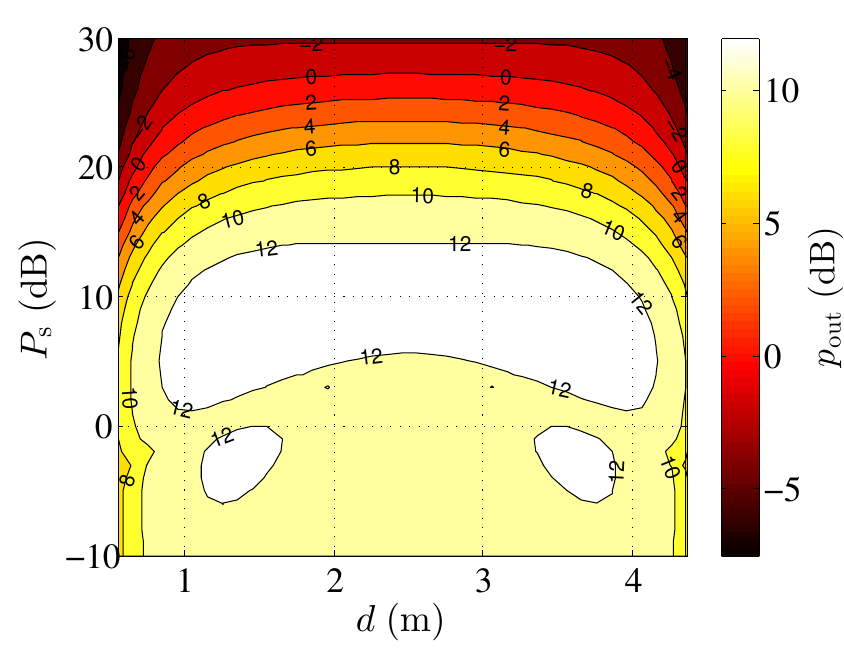}
\caption{OP performance in dB with respect to varying RD values $d$ in meters and DS power allocation values $P_{\rm s}$ in dB.}
\label{Fig_5_validation_plot_contour}
\end{figure}
The role of the parameters $d$ and $P_{\rm s}$ for the RD and the DS power allocation, respectively, in the OP performance of the proposed RIS-aided underlaid D2D communication system is depicted in Fig.~\ref{Fig_5_validation_plot_contour}. The considered values for $d$ and $P_{\rm s}$ are plotted in the $x$- and $y$-axis, respectively, whereas $p_{\rm out}$ values are shown in the contour. As illustrated, for a fixed $d$ value, OP improves with increasing $P_{\rm s}$; this behaviour confirms the claim in Lemma~\ref{lemma_decouple_vars}. 
Additionally, it is shown that, for a fixed $P_{\rm s}$ value, there exist two minima points at the OP, which appear at the corner points of the horizontal axis with the $d$ values; these points are the ones given by the constraints 
$\textstyle({\rm C4})$ and $\textstyle({\rm C5})$. The figure also verifies that the proposed global optimum is the same with the simulated one, thus, validating the 
claims expressed via expressions \eqref{eq_d_opt_parallel} and \eqref{eq_Ps_opt}.

\begin{figure}[!t]
\centering
\includegraphics[trim=3cm 9cm 3.8cm 9.5cm, width=0.9\columnwidth, clip]{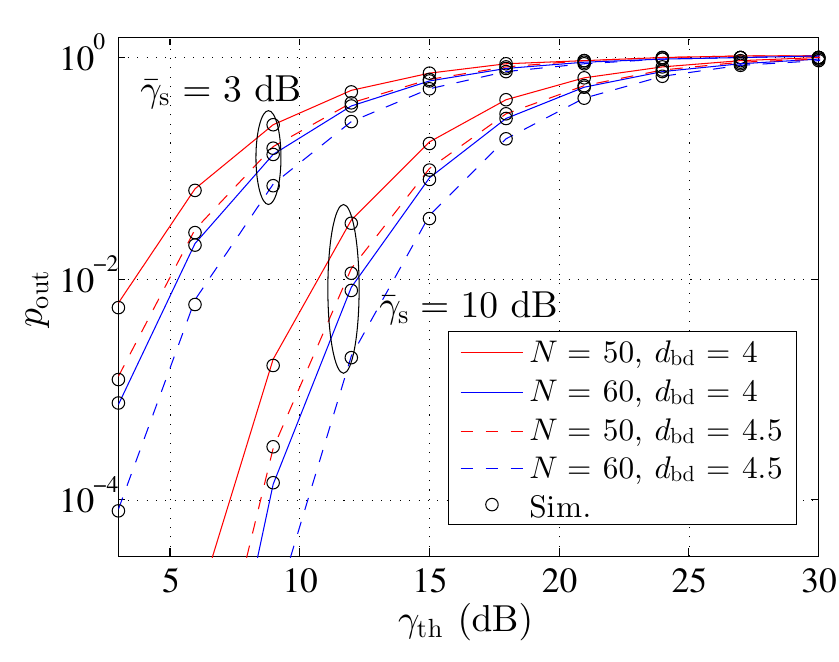}
\caption{OP performance as a function of the SINR threshold $\gamma_{\rm {th}}$ in dB for $d_{\rm bd}= \{4, 4.5\}$ m, $N=\{50, 60\}$, and $\bar{\gamma}_{\rm s} = \{3, 10\}$ dB.}
\label{Fig_6_validation_plot_OP}
\end{figure}
In Fig.~\ref{Fig_6_validation_plot_OP}, we present the OP performance of the proposed D2D system, via the numerical evaluation of \eqref{eq_out_d2d} in Section~\ref{subsec_OP}, as a function of the SINR threshold $\gamma_{\rm {th}}$, considering that $N=\{50, 60\}$, $\bar{\gamma}_{\rm s} = \{3, 10\}$ dB, 
and $d_{\rm bd}= \{3.5, 4.5\}$ m. It is depicted that, for fixed BS-DU link distance $d_{\rm bd}$ and $\gamma_{\rm th}$, the OP improves with increasing $N$. However, for given $N$ and $\gamma_{\rm th}$, a smaller $d_{\rm bd}$ results in increased interference from BS, thus, degrading the OP at DU. Similarly, for a fixed $\gamma_{\rm th}$, when the average SNR at DS, $\bar{\gamma}_{\rm s}$, increases, the OP improves.  

\subsection{Design Insights and Comparisons} 
\begin{figure}[!t]
\centering
\includegraphics[trim=3cm 9cm 3.8cm 9.5cm, width=0.9\columnwidth, clip]{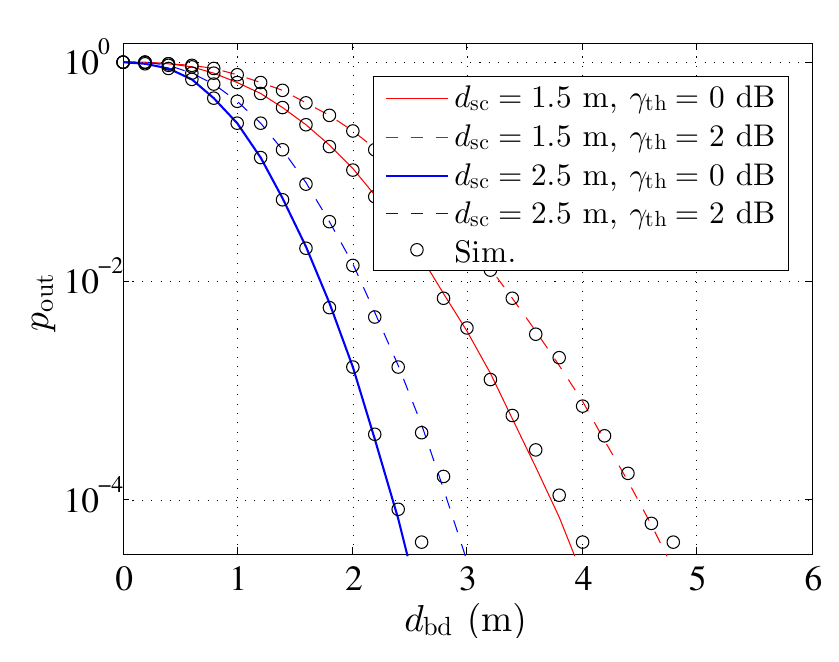}
\caption{OP performance as a function of the BS-DU distance $d_{\rm bd}$ in meters for $d_{\rm sc}=\{1.5, 2.5\}$ m and $\gamma_{\rm {th}}=\{0, 2\}$ dB.}
\label{Fig_7_insight_plot_dSC}
\end{figure}
The following Figs.~\ref{Fig_7_insight_plot_dSC} and~\ref{Fig_7_insight_plot_dSD_dBD} illustrate the OP performance versus the BS-DU link distance $d_{\rm bd}$ and the RD value $d$, respectively. In Fig.~\ref{Fig_7_insight_plot_dSC}, the OP is investigated with respect to varying DS-CU link distance $d_{\rm sc}$ and $\gamma_{\rm th}$, while Fig.~\ref{Fig_7_insight_plot_dSD_dBD} examines the role of the common RIS element amplitude $\alpha$ and $\gamma_{\rm th}$. It can be observed from the former figure that, for given $d_{\rm sc}$ and $\gamma_{\rm th}$, OP improves with increasing $d_{\rm bd}$. Note that, as $d_{\rm bd}$ increases, the pathloss becomes more severe. This results in a weaker interference level at DU, hence, a lower $d_{\rm bd}$ value is obtained. In addition, for fixed $d_{\rm bd}$ and $\gamma_{\rm th}$, OP decreases with increasing $d_{\rm sc}$. This happens because a larger $d_{\rm sc}$ value results in a reduced interference level $I_{\rm b}$. It can be also seen that, for a given $d_{\rm sc}$, a decrease in $\gamma_{\rm th}$ results in OP improvement. 

\begin{figure}[!t]
\centering
\includegraphics[trim=3cm 9cm 3.8cm 9.5cm, width=0.9\columnwidth, clip]{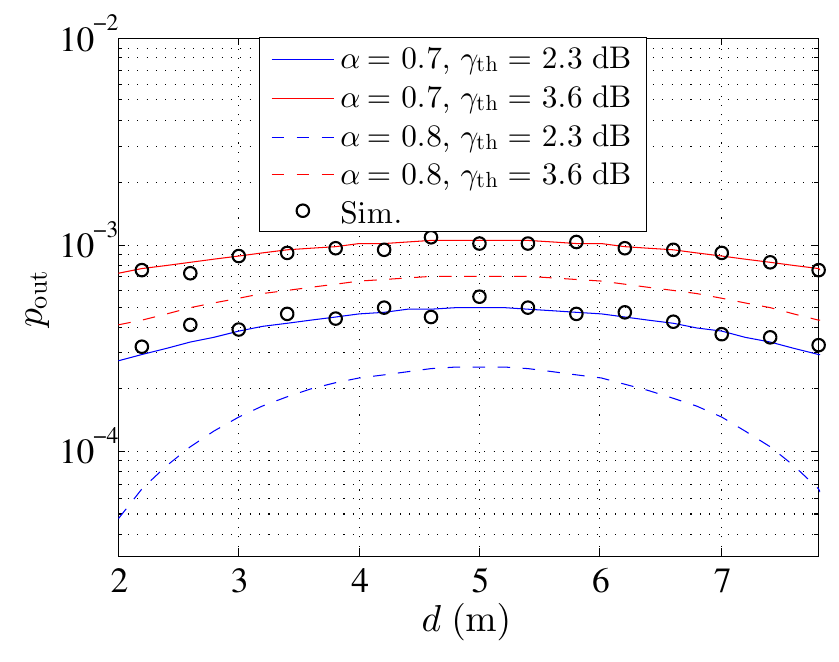}
\caption{OP performance as a function of the RD value $d$ in meters for $\gamma_{\rm {th}}=\{2.3, 3.6\}$ dB and $\alpha=\{0.7, 0.8\}$.}
\label{Fig_7_insight_plot_dSD_dBD}
\end{figure}
Figure~\ref{Fig_7_insight_plot_dSD_dBD} showcases that, for a given $\gamma_{\rm {th}}$, OP improves as $\alpha$ increases. This indicates that a larger reflection coefficient contributes in the quality of the received signal at DU. Moreover, it can be observed that, for a given $\alpha$, the OP performance improves with decreasing $\gamma_{\rm th}$. It is finally shown in this figure that RIS needs to be placed at the corner points to yield minimum OP. This result verifies the finding in \eqref{eq_d_opt_parallel}. 

\begin{figure}[!t]
\centering
\includegraphics[trim=2.9cm 9cm 3.8cm 9.5cm, width=0.9\columnwidth, clip]{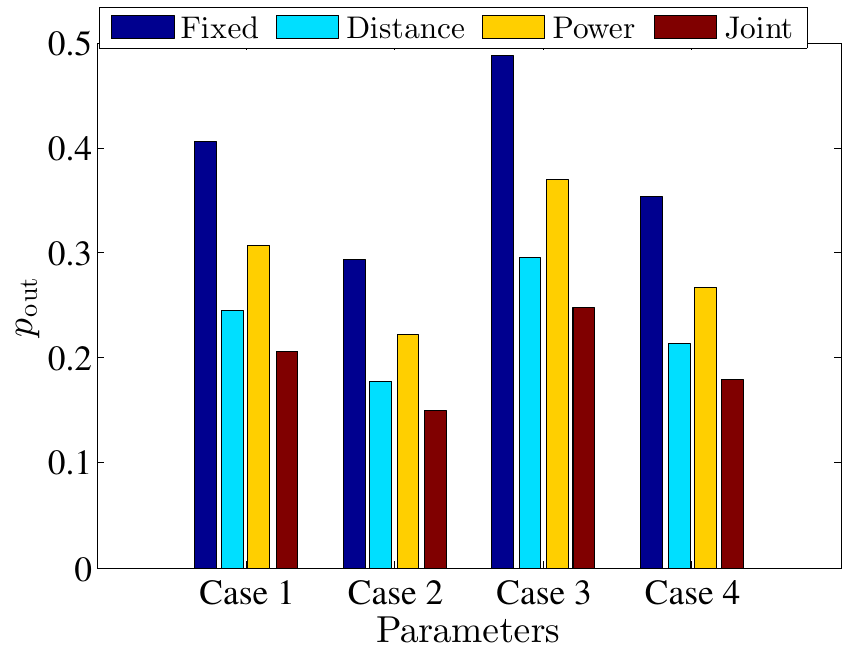}
\caption{OP performance comparison between the proposed jointly optimal scheme and benchmark ones (optimal power allocation with fixed RD and optimal RD with fixed power) for $\gamma_{\rm {th}}=\{0, 2\}$ dB and
$N=\{40, 60\}$.}
\label{FIG_comp_plot_OPT}
\end{figure}
In Fig.~\ref{FIG_comp_plot_OPT}, we compare the OP performance of the proposed jointly optimal design with two semi-adaptive schemes: the optimal power allocation with fixed RD and the optimal RD with fixed power. The performance curves were obtained via the numerical evaluation of \eqref{eq_d_opt_parallel} and \eqref{eq_Ps_opt} for all three schemes and for different combinations of $N$ and $\gamma_{\rm th}$, namely: 
$N=40$ and $\gamma_{\rm {th}}$ = $0$ dB (Case 1); $N=60$ and $\gamma_{\rm {th}}=0$ dB (Case 2); $N=40$ and $\gamma_{\rm {th}}=2$ dB (Case 3); and $N=60$ and $\gamma_{\rm {th}}=2$ dB (Case 4). It can be observed, for our proposed scheme, an improvement of $28$\% from Case 1 to 2, $27$\% improvement from Case 3 to 4, and $16$\% improvement from Case 4 to 2.  
By comparing Cases 1 and 2, it is shown that, for a fixed $\gamma_{\rm th}$, the OP performance degrades as $N$ increases. On the other hand, for a fixed $N$ value, OP improves as $\gamma_{\rm th}$ increases. It is finally showcased that, between the two benchmark schemes, the optimal distance one fares always better than the optimal power one, which is, however, always outperformed by the proposed jointly optimal design. The improvement is $44\%$ and $20\%$ compared to the optimal power and optimal distance schemes, respectively. 

\begin{figure}[!t]
\centering
\includegraphics[trim=2.9cm 9cm 3.8cm 9.5cm, width=0.9\columnwidth, clip]{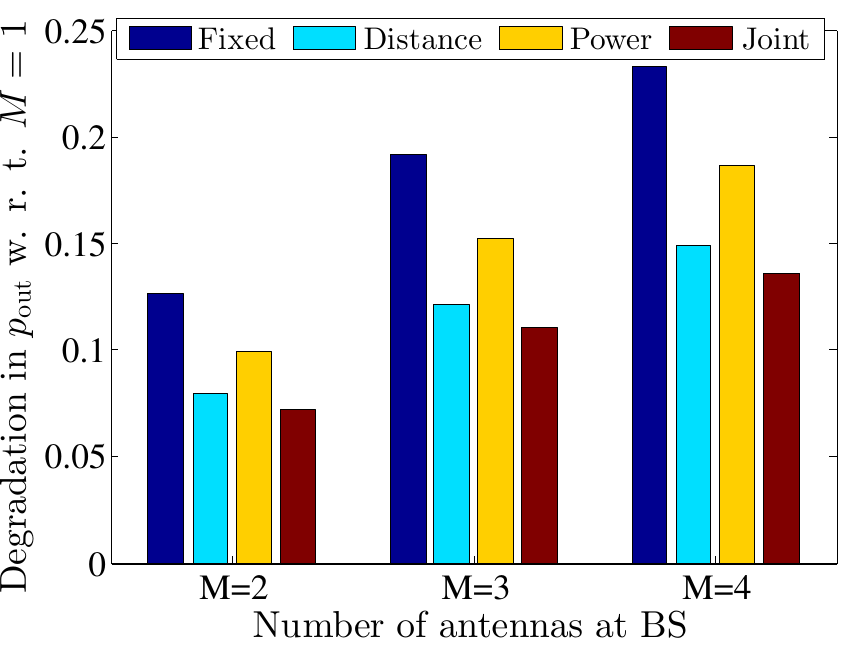}
\caption{OP performance comparison between the joint, and benchmark schemes 
(optimal power allocation with fixed RD and optimal RD with fixed power) for 
$M$ = $\{2, 3, 4\}$ with respect to $M=1$ case.}
\label{FIG_comp_plot_num_ant}
\end{figure} In Fig. \ref{FIG_comp_plot_num_ant}, we plot OP using \eqref{eq_d_opt_parallel}, 
\eqref{eq_Ps_opt} offered by the three optimization schemes for different number of antennas $M$ at the BS. It can be seen that the OP performance degrades with an increase in the number of antennas. As the number of antennas at the primary transmitter increases, the SNR at the primary receiver gets enhanced. However, this is seen as an interfering effect at the D2D receiver. Therefore, the throughput at the receiver decreases. Typically, when $M=4$, we observe a $15\%$, $18\%$ and $13\%$ increase in OP for optimal RD with fixed power, optimal power allocation with fixed RD, and jointly optimal schemes, respectively, as compared to $M=1$ case.

A fair comparison of the present method over the works \cite{Duong2022_RIS_UAV,Poor2020_improper,Q_Wu_JSAC_2020} is not 
relevant as the RIS deployment optimizes different objectives in each work and mathematical tool sets used for solving the 
optimization problems are different. 

\section{Conclusions}
\label{sec_conclude}
In this paper, we presented a novel OP minimization problem for an RIS-empowered underlaid D2D communication system. The proposed problem targeted the jointly optimal design of the power budget at the D2D source node and the RIS deployment. Based on an efficient Taylor series approximation, we presented a closed-form expression for the average SINR at the D2D destination node under Rayleigh fading conditions, which served as an analytical equivalent of the OP performance. 
We demonstrated how the restriction on the interference power and the total power budget affect the optimal power. 
It was also indicated that the optimal location of the RIS occurs at its corner points, as they are given by the respective distance constraints. Our extensive numerical results validated our derived analytical expressions and the optimized system parameters, 
namely, the global optimal distance of the RIS location and the D2D source transmit power. 
It was demonstrated that the proposed jointly optimal scheme can yield $44\%$ and $20\%$ improvement over the considered semi-adaptive optimal power and optimal distance benchmark schemes, respectively. 
In the future, we intend to incorporate multiple D2D pairs into the system model and conduct relevant mathematical analysis, considering also more sophisticated precoding schemes at the BS, like maximal-ratio transmission and zero forcing. Additionally, we will study the secrecy performance when there exist numerous eavesdroppers in a legitimate RIS-empowered underlaid D2D system. 
\bibliographystyle{IEEEtran}
\bibliography{main}
\end{document}